\documentclass[a4paper, 11pt]{article}
\usepackage[utf8]{inputenc}
\usepackage[T1]{fontenc}
\usepackage[english]{babel}

\usepackage{times, amssymb, amsthm, mathtools}
\usepackage{algorithm, algpseudocode}
\usepackage[hidelinks]{hyperref}
\usepackage{enumitem}
\usepackage{multirow}
\usepackage{siunitx} 

\usepackage[authoryear, sort, comma, longnamesfirst]{natbib}
\emergencystretch=\maxdimen
\numberwithin{equation}{section} 
\counterwithin{figure}{section}
\counterwithin{table}{section}

\newtheorem{theorem}{Theorem}[section]
\newtheorem{lemma}[theorem]{Lemma}
\newtheorem{remark}{Remark}[section]
\theoremstyle{definition}
\newtheorem{definition}{Definition}[section]
\newtheorem{example}{Example}[section]
\newtheorem{problem}{Problem}[section]

\algrenewcommand{\alglinenumber}[1]{\footnotesize Step #1:}

\newcommand{\ind}[1]{\mathrlap{\qquad #1}} 
\newcommand{\texteq}{\mathrm}
\newcommand{\card}[1]{\lvert #1 \rvert} 
\let\proglang=\textsf 

\def\tmin{$\square$}
\def\tmax{$\blacksquare$}

\def\N{{\mathbb N}}
\def\R{{\mathbb R}}

\def\A{\mathcal A}
\def\B{\mathcal B}
\def\C{\mathcal C}
\def\D{\mathcal D}
\def\L{\mathcal L}
\def\U{\mathcal U}
\def\H{\mathcal H}

\def\x{\mathbf x}
\def\n{\mathbf n}
\def\zero{\mathbf 0}

\def\Lt{\widetilde{\mathcal L}}
\def\Ut{\widetilde{\mathcal U}}

\title{Recursive Neyman Algorithm for Optimum Sample Allocation under Box Constraints on Sample Sizes in Strata}
\author{
    Jacek Wesołowski
    \thanks{Programming, Coordination of Statistical Surveys and Registers Department, Statistics Poland, Aleja 
    Niepodległości 208, 00-925 Warsaw, Poland, and Faculty of Mathematics and Information Science, Warsaw University of 
    Technology, ul. Koszykowa 75, 00-662 Warsaw, Poland. E-mail: \mbox{jacek.wesolowski@pw.edu.pl}},
    Robert Wieczorkowski
    \thanks{Programming, Coordination of Statistical Surveys and Registers Department, Statistics Poland, Aleja Niepodległości 208, 00-925 Warsaw, Poland. \mbox{E-mail: R.Wieczorkowski@stat.gov.pl}},
    Wojciech Wójciak
    \thanks{Faculty of Mathematics and Information Science, Warsaw University of Technology, ul. Koszykowa 75, 00-662 Warsaw, Poland. E-mail: \mbox{wojciech.wojciak.dokt@pw.edu.pl}}
}

\begin{document}
\maketitle
	
\begin{abstract}
The optimum sample allocation in stratified sampling is one of the basic issues of survey methodology. It is a procedure of 
dividing the overall sample size into strata sample sizes in such a way that for given sampling designs in strata the variance of 
the {\em stratified $\pi$ estimator} of the population total (or mean) for a given study variable assumes its minimum. In this 
work, we consider the optimum allocation of a sample, under lower and upper bounds imposed jointly on sample sizes in 
strata. We are concerned with the variance function of some generic form that, in particular, covers the case of the {\em 
simple random sampling without replacement} in strata.
The goal of this paper is twofold. First, we establish (using the Karush-Kuhn-Tucker conditions) a generic form of the optimal 
solution, the so-called optimality conditions. Second, based on the established optimality conditions, we derive an efficient 
recursive algorithm, named {\em RNABOX}, which solves the allocation problem under study. The {\em RNABOX} can be 
viewed as a generalization of the classical recursive Neyman allocation algorithm, a popular tool for optimum allocation when 
only upper bounds are imposed on sample strata-sizes. We implement {\em RNABOX} in \proglang{R} as a part of our 
package {\tt stratallo} which is available from the Comprehensive R Archive Network (CRAN) repository.

\end{abstract}

\section{Introduction} 
\label{sec:intro}

Let us consider a finite population $U$ of size $N$. Suppose the parameter of interest is the population total $t$ of a variable 
$y$ in $U$, i.e. $t = \sum_{k \in U}\, y_k$, where $y_k$ denotes the value of $y$ for population element $k \in U$. To 
estimate $t$, we consider the {\em stratified sampling} with the {\em $\pi$ estimator}. Under this well-known sampling 
technique, population $U$ is stratified, i.e. $U = \bigcup_{h \in \H}\, U_h$, where $U_h,\, h \in \H$, called strata, are disjoint 
and non-empty, and $\H$ denotes a finite set of strata labels. The size of stratum $U_h$ is denoted $N_h,\, h \in \H$ and 
clearly $\sum_{h \in \H} N_h = N$. Probability samples $s_h \subseteq U_h$ of size $n_h \leq N_h,\, h \in \H$, are selected 
independently from each stratum according to chosen sampling designs which are often of the same type across strata. The 
resulting total sample is of size $n = \sum_{h \in \H}\, n_h \leq N$. It is well known that the {\em stratified $\pi$ estimator} 
$\hat{t}_{\pi}$ of $t$ and its variance are expressed in terms of the first and second order inclusion probabilities \citep[see, 
e.g.][Result 3.7.1, p.~102]{Sarndal}. In particular, for several important sampling designs
\begin{equation}
    \label{eq:var}
    Var(\hat{t}_{\pi}) = \sum_{h \in \H}\, \tfrac{A_h^2}{n_h} - B,
\end{equation}
where $A_h > 0,\, B$ do not depend on $n_h,\, h \in \H$.
Among the most basic and common sampling designs that give rise to the variance of the form \eqref{eq:var} is the {\em 
simple random sampling without replacement} in strata (abbreviated {\em STSI}). In this case, the {\em stratified $\pi$ 
estimator} of $t$ assumes the form
\begin{equation}
    \hat{t}_{\pi} = \sum_{h \in \H}\, \tfrac{N_h}{n_h} \sum_{k \in s_h} y_k,
\end{equation}
which yields in \eqref{eq:var}: $A_h = N_h S_h$, where $S_h$ denotes stratum standard deviation of study variable $y$, $h 
\in \H$, and $B = \sum_{h \in \H}\, N_h S_h^2$ \citep[see, e.g.][Result 3.7.2, p.~103]{Sarndal}. 

The classical problem of optimum sample allocation is formulated as the determination of the allocation vector $\n = (n_h,\, h 
\in \H)$ that minimizes the variance \eqref{eq:var}, subject to $\sum_{h \in \H}\, n_h = n$, for a given $n \leq N$ \citep[see, 
e.g.][Section 3.7.3, p.~104]{Sarndal}. In this paper, we are interested in the classical optimum sample allocation problem with 
additional two-sided constraints imposed on sample sizes in strata. We phrase this problem in the language of mathematical 
optimization as Problem \ref{prob}.

\begin{problem}
    \label{prob}
    Given a finite set $\H \neq \emptyset$ and numbers $A_h > 0,\,m_h,\, M_h,\, n$, such that $0<m_h < M_h \leq N_h,\, h \in 
    \H$ and $\sum_{h \in \H}\, m_h \le n \le \sum_{h \in \H}\, M_h$,	
    \begin{align}
		\underset{\x\, =\, (x_h,\, h\, \in\, \H)\, \in\, \R_+^{\card{\H}}}{\texteq{minimize ~\,}} & \quad \sum_{h \in \H} 
		\tfrac{A_h^2}{x_h} \label{eq:prob_obj} \\
		\texteq{subject ~ to \quad\,\,\,} & \quad \sum_{h \in \H} x_h = n \label{eq:prob_cnst_eq} \\
		& \quad m_h \le x_h \le M_h,  \ind{h \in \H.} \label{eq:prob_cnst_ineq}
	\end{align}
\end{problem}

To emphasize the fact that the optimal solution to Problem \ref{prob} may not be an integer one, we denote the optimization variable by $\x$, not by $\n$. The assumptions about $n,\, m_h,\, M_h,\, h \in \H$, ensure that Problem \ref{prob} is feasible. 

The upper bounds imposed on $x_h,\, h \in \H$, are natural since for instance the solution with $x_h > N_h$ for some $h \in 
\H$ is impossible. The lower bounds are necessary e.g. for estimation of population strata variances $S_h^2,\, h \in \H$. They 
also appear when one treats strata as domains and assigns upper bounds for variances of estimators of totals in domains. 
Such approach was considered e.g. in \citet{Choudhry}, where apart of the upper bounds constraints $x_h \le N_h,\, h \in \H$, 
the additional constraints $\left(\tfrac{1}{x_h} - \tfrac{1}{N_h} \right) N_h^2 S_h^2 \le R_h,\, h \in \H$, where $R_h,\, h \in \H$, 
are given constants, have been imposed. Obviously, the latter system of inequalities can be rewritten as lower bounds 
constraints of the form $x_h \geq m_h = \tfrac{N_h^2 S_h^2}{R_h + N_h S_h^2},\, h \in \H$. The solution given in 
\citet{Choudhry} was obtained by the procedure  based on the Newton-Raphson algorithm, a general-purpose root-finding 
numerical method. See also a related paper by \citet{WNB}, where the optimum allocation problem under the constraint of the 
equal precision for estimation of the strata means was considered.

\bigskip
It is convenient to introduce the following definition for feasible solutions of Problem \ref{prob}.

\begin{definition}
    \label{def:vertex_regular}
    {\em Any vector $\x = (x_h,\, h \in \H)$ satisfying \eqref{eq:prob_cnst_eq} and \eqref{eq:prob_cnst_ineq} will be called an {\em allocation}. 

    An allocation  $\x= (x_h,\, h \in \H)$ is called a {\em vertex} one if and only if
    \begin{equation*}
		x_h = 
		\begin{cases}
			m_h, & \ind{h \in \L} \\
			M_h, & \ind{h \in \U,}
		\end{cases}
    \end{equation*}
    where $\L,\, \U \subseteq \H$ are such that $\L \cup \U = \H$ and $\L \cap \U = \emptyset$. 
 
    An allocation which is not a {\em vertex} one  will be called a {\em regular} allocation.
    
    The solution to Problem \ref{prob} will be called the {\em optimum} allocation. }
\end{definition}

Note that an optimum allocation may be of a {\em vertex} or of a {\em regular} form. The name {\em vertex} allocation refers 
to the fact that in this case $\x$ is a vertex of the hyper-rectangle $\times_{h\in\H}\,[m_h,\, M_h]$. We note that Problem 
\ref{prob} becomes trivial if $n = \sum_{h \in \H}\, m_h$ or $n = \sum_{h \in \H}\, M_h$. In the former case, the solution is 
$\x^* = (m_h,\, h \in \H)$, and in the latter $\x^* = (M_h,\, h \in \H)$. These two are boundary cases of the {\em vertex} 
allocation. In real surveys with many strata, a {\em vertex} optimum allocation rather would not be expected. Nevertheless, for 
completeness we also consider such a case in Theorem \ref{th:optcond}, which describes the form of the optimum allocation 
vector. We also note that a {\em regular} optimum allocation $\x^* \in \times_{h \in \H}\, (m_h,\, M_h)$ if and only if it is the 
classical Tschuprow-Neyman allocation $\x^*= (A_h\, \tfrac{n}{\sum_{v \in \H} A_v},\, h \in \H)$ 
\citep[see][]{Tschuprow1923b, Neyman}.

\bigskip
The rest of this paper is structured as follows. Section \ref{sec:motivation} presents motivations for this research as well as a 
brief review of the literature. In Section \ref{sec:optcon}, we identify Problem \ref{prob} as a convex optimization problem and 
then use the Karush-Kuhn-Tucker conditions to establish necessary and sufficient conditions for a solution to optimization 
Problem \ref{prob}. These conditions, called the optimality conditions, are presented in Theorem \ref{th:optcond}. In Section 
\ref{sec:rnabox}, based on these optimality conditions, we introduce a new algorithm, {\em RNABOX}, and prove that it solves 
Problem \ref{prob} (see Theorem \ref{th:rnabox}). The name {\em RNABOX} refers to the fact that this algorithm generalizes 
the recursive Neyman algorithm, denoted here {\em RNA}. The {\em RNA} is a well-established allocation procedure, 
commonly used in everyday survey practice. It finds a solution to the allocation Problem \ref{prob:upper} (see below), which 
is a relaxed version of Problem \ref{prob}. 
As we shall see in Section \ref{sec:naive_rnabox}, a naive modification of the recursive Neyman  approach, which works fine 
for the one-sided upper bounds imposed on the sample strata sizes, does not give the correct solution in the case of 
two-sided bounds. A more subtle approach is needed.

\bigskip
Finally, let us note that the implementation of {\em RNABOX} algorithm is available through our \proglang{R} package {\tt 
stratallo} \citep{stratallo}, which is published in CRAN repository \citep{R}.

\section{Motivation and literature review}
\label{sec:motivation}

An abundant body of literature is devoted to the problem of optimum sample allocation, going back to classical solution of 
\citet{Tschuprow1923b} and \citet{Neyman}, dedicated to {\em STSI} sampling without taking inequality constraints 
\eqref{eq:prob_cnst_ineq} into account. In spite of this fact, a thorough analysis of the literature shows that Problem 
\ref{prob} has not been completely understood yet and it suffers from the lack of fully satisfactory algorithms. 

Below, we briefly review existing methods for solving Problem \ref{prob}, including methods that provide integer-valued 
solutions.

\subsection{Not-necessarily integer-valued allocation} 
\label{ssec:motivation_notint}

An approximate solution to Problem \ref{prob} can be achieved through generic methods of non-linear programming (NLP) 
\citep[see, e.g. the monograph][and references therein]{Valliant}. These methods have been involved in the problem of 
optimum sample allocation since solving the allocation problem is equivalent to finding the extreme (namely, stationary 
points) of a certain objective function over a feasible set. Knowing the (approximate) extreme of the objective function, one 
can determine the (approximate, yet typically sufficiently accurate) sizes of samples allocated to individual strata. 

In a similar yet different approach adopted e.g. in \citet{MSW}, Problem \ref{prob} is transformed into root-finding or 
fixed-point-finding problems (of some properly defined function) to which the solution is obtained by general-purpose 
algorithms like e.g. {\em bisection} or {\em regula falsi}.

Algorithms used in both these approaches would in principle have infinitely many steps, and are stopped by an arbitrary 
decision, typically related to the precision of the iterates. There are two main weaknesses associated with this way of 
operating: failure of the method to converge or slow convergence towards the optimal solution for some poor starting points. 
In other words, performance of these algorithms may strongly depend on an initial choice of a starting point, and such a 
choice is almost always somewhat hazardous. These and similar deficiencies are discussed in details in Appendix 
\ref{app:msw} in the context of the allocation methods described in \citet{MSW}. Another drawback of the algorithms of this 
type is their sensitivity to finite precision arithmetic issues that can arise in case when the stopping criterion is not expressed 
directly in terms of the allocation vector iterates (which is often the case).

Contrary to that, in the recursive algorithms (we are concerned with), the optimal solution is always found by recursive search 
of feasible candidates for the optimum allocation among subsets of $\H$. Hence, they stop always at the exact solution and 
after finitely many iterations (not exceeding the number of strata + 1, as we will see for the case of {\it RNABOX} in the proof 
of Theorem \ref{th:rnabox}). An important example of such an algorithm, is the recursive Neyman algorithm, {\em RNA}, 
dedicated for Problem \ref{prob:upper}, a relaxed version of Problem \ref{prob}.
\begin{problem}
    \label{prob:upper}
    Given a finite set $\H \neq \emptyset$ and numbers $A_h > 0,\, M_h,\, n > 0$, such that $0<M_h \leq N_h,\, h \in \H$ and $n \le \sum_{h \in \H}\, M_h$, 
    \begin{align*}
		\underset{\x\, =\, (x_h,\, h\, \in\, \H)\, \in\, \R_+^{\card{\H}}}{\texteq{minimize ~\,}}  & \quad \sum_{h \in \H} 
		\tfrac{A_h^2}{x_h} \\
		\texteq{subject ~ to \quad\,\,\,}	& \quad \sum_{h \in \H} x_h = n \\
		& \quad x_h \le M_h,  \ind{h \in \H.}
    \end{align*}
\end{problem}
Although {\em RNA} is popular among practitioners, a formal proof of the fact that it gives the optimal solution to Problem 
\ref{prob:upper} has been given only recently in \citet{WWW}. For other recursive approaches to Problem \ref{prob:upper}, 
see also e.g. \citet{SG}, \citet{Kadane}.

To the best of our knowledge, the only non-integer recursive optimum allocation algorithm described in the literature that is 
intended to solve Problem \ref{prob} is the \emph{noptcond} procedure proposed by \citet{GGM}. In contrary to {\em 
RNABOX}, this method in particular performs strata sorting. Unfortunately, the allocation computed by {\em noptcond} may 
not yield the minimum of the objective function \eqref{eq:prob_obj}; see Appendix \ref{app:ggm} for more details.

\subsection{Integer-valued allocation} 
\label{ssec:motivation_int}

Integer-valued algorithms dedicated to Problem \ref{prob} are proposed in \cite{Friedrich, Wright2017, Wright2020}. The 
multivariate version of the optimum sample allocation problem under box constraints in which $m_h = m,\, h \in \H$, for a 
given constant $m$, is considered in the paper of \citet{Brito}. The proposed procedure that solves that problem uses binary 
integer programming algorithm and can be applied to the univariate case. See also \cite{stratbr} for the 
\proglang{R}-implementation of this approach.

Integer-valued allocation methods proposed in these papers are precise (not approximate) and theoretically sound. However, 
they are relatively slow, when compared with not-necessarily integer-valued algorithms. For instance, at least for one-sided 
constraints, the integer-valued algorithm {\em capacity scaling} of \citet{Friedrich} may be thousands of times slower than 
the {\it RNA} \citep[see][Section 4]{WWW}. This seems to be a major drawback of these methods as the differences in 
variances of estimators based on integer-rounded non-integer optimum allocation and integer optimum allocation are 
negligible as explained in Section \ref{sec:conclusions}. The computational efficiency is of particular significance when the 
number of strata is large, see, e.g. application to the German census in \citet{Burgard}, and it becomes even more 
pronounced in iterative solutions to stratification problems, when the number of iterations may count in millions \citep[see, 
e.g.][]{Lednicki2, GH, KNA, Baillargeon, Barcaroli}.

\bigskip
Having all that said, the search for a new, universal, theoretically sound and computationally effective recursive algorithms of optimum sample allocation under two-sided constraints on the sample strata-sizes, is crucial both for theory and practice of survey sampling.

\section{Optimality conditions} 
\label{sec:optcon}

In this section we establish optimality conditions, that is, a general form of the solution to Problem \ref{prob}. As it will be 
seen in Section \ref{sec:rnabox}, these optimality conditions are crucial for the construction of {\em RNABOX} algorithm.

Before we establish necessary and sufficient optimality conditions for a solution to optimization Problem \ref{prob}, we first define a set function $s$, which considerably simplifies notation and  calculations.

\begin{definition}
    \label{def:s}
    Let $\H,\, n,\, A_h > 0,\, m_h,\, M_h,\, h \in \H$ be as in Problem \ref{prob} and let $\L,\, \U \subseteq \H$ be such that $\L 
    \cap \U = \emptyset,\, \L \cup \U \subsetneq \H$. The set function $s$ is defined as
    \begin{equation}
        \label{eq:s}
		s(\L,\, \U) = \frac{n - \sum_{h \in \L} m_h - \sum_{h \in \U} M_h}{\sum_{h \in \H \setminus (\L \cup \U)} A_h}.
    \end{equation}
\end{definition}

Below, we will introduce the $\x^{(\L,\, \U)}$ vector for disjoint $\L,\, \U \subseteq \H$. It appears that the solution of the 
Problem \ref{prob} is necessarily of the form \eqref{eq:LU} with sets $\L$ and $\U$ defined implicitly through systems of 
equations/inequalities established in Theorem \ref{th:optcond}.

\begin{definition}
    \label{def:LU}
    Let $\H,\, n,\, A_h > 0,\, m_h,\, M_h,\, h \in \H$ be as in Problem \ref{prob}, and let $\L,\, \U \subseteq \H$ be such that $\L 
    \cap \U = \emptyset$. We define the vector $\x^{(\L,\, \U)} = (x_h^{(\L,\, \U)},\, h \in \H)$ as follows
    \begin{equation}
		\label{eq:LU}
		x_h^{(\L,\, \U)} = 
		\begin{cases}
			m_h,					& \ind{h \in \L} \\
			M_h,					& \ind{h \in \U} \\
			A_h\, s(\L,\, \U),	& \ind{h \in \H \setminus (\L \cup \U).}
		\end{cases}
    \end{equation}	
\end{definition}

The following Theorem \ref{th:optcond} characterizes the form of the optimal solution to Problem \ref{prob} and therefore is 
one of the key results of this paper.

\begin{theorem}[Optimality conditions]
    \label{th:optcond}
    The optimization Problem \ref{prob} has a unique optimal solution. Point $\x^* \in \R_+^{\card{\H}}$ is a solution to 
    optimization Problem \ref{prob} if and only if $\x^*= \x^{(\L^*,\, \U^*)}$, with disjoint $\L^*,\, \U^* \subseteq \H$, such that 
    one of the following two cases holds:
    \begin{enumerate}[wide, labelindent=0pt, leftmargin=*]
	\item[CASE I:] $\L^* \cup \U^* \subsetneq \H$ and
	\begin{equation}
		\label{eq:optcond}
		\begin{gathered}
			\L^* = \left\{h \in \H:\, s(\L^*,\, \U^*) \leq \tfrac{m_h}{A_h} \right\}, \\
			\U^* = \left\{h \in \H:\, s(\L^*,\, \U^*) \geq \tfrac{M_h}{A_h} \right\}.
		\end{gathered}
	\end{equation}
	
	\item[CASE II:] $\L^* \cup \U^* = \H$ and
	\begin{gather}
		\max_{h \in \U^*} \tfrac{M_h}{A_h} \leq \min_{h \in \L^*} \tfrac{m_h}{A_h} \ind{\text{ if\, } \U^* \neq \emptyset \text{ and } 
		\L^* \neq \emptyset,}, \label{eq:optcond_vertex} \\
		\sum_{h \in \L^*} m_h + \sum_{h \in \U^*} M_h = n. \label{eq:optcond_vertex_n}
	\end{gather}
    \end{enumerate}
\end{theorem}

\begin{remark}
	\label{rem:optcond:reg_ver}
    The optimum allocation $\x^*$ is a {\em regular} one in CASE I and a {\em vertex} one in CASE II.
\end{remark}

The proof of Theorem \ref{th:optcond} is given in Appendix \ref{app:optcon}. Note that Theorem \ref{th:optcond} describes 
the general form of the optimum allocation up to specification of {\em take-min} and {\em take-max} strata sets $\L^*$ and 
$\U^*$. The question how to identify sets $\L^*$ and $\U^*$ that determine the optimal solution $\x^* = \x^{(\L^*,\, \U^*)}$ 
is the subject of Section \ref{sec:rnabox}.

\section{Recursive Neyman algorithm under box constraints} 
\label{sec:rnabox}

\subsection{The {\em RNABOX} algorithm} 
\label{ssec:rnabox}

In this section we introduce an algorithm solving Problem \ref{prob}. In view of Theorem \ref{th:optcond} its essential task is 
to split the set of all strata labels $\H$ into three subsets of {\em take-min} ($\L^*$), {\em take-max} ($\U^*$), and {\em 
take-Neyman} ($\H \setminus (\L^* \cup \U^*)$). We call this new algorithm {\em RNABOX} since it generalizes existing 
algorithm {\em RNA} in the sense that {\em RNABOX} solves optimum allocation problem with simultaneous lower and upper 
bounds, while the {\em RNA} is dedicated for the problem with upper bounds only, i.e. for Problem \ref{prob:upper}. 
Moreover,  {\em RNABOX} uses {\em RNA} in one of its interim steps. We first recall {\em RNA} algorithm and then present 
{\em RNABOX}. For more information on {\em RNA}, see \citet[Section 2]{WWW} or \citet[Remark 12.7.1, p.~466]{Sarndal}.

\begin{algorithm}[ht]
    \caption{{\em RNA}}
    \label{alg:rna}
    \textbf{Input:} $\H,\, (A_h)_{h \in \H},\, (M_h)_{h \in \H},\, n$.
    \begin{algorithmic}[1]
		\Require $A_h > 0,\, M_h > 0,\, h \in \H,\, 0 < n \leq \sum_{h \in \H} M_h$.
		\State Set $\U = \emptyset$.
		\State\label{alg:rna:U}Determine $\Ut = \left\{h \in \H \setminus \U:\, A_h\, s(\emptyset,\, \U) \ge M_h\right\}$, where set 
		function $s$ is defined in \eqref{eq:s}.
        \State If {$\Ut = \emptyset$}, go to \footnotesize Step \ref{alg:rna:step_return} \normalsize. Otherwise, update $\U \gets 
        \U \cup \Ut$ and go to \footnotesize Step \ref{alg:rna:U} \normalsize.
        \State\label{alg:rna:step_return}Return $\x^* = (x^*_h,\, h \in \H)$ with
		$x^*_h =
		\begin{cases}
			M_h,							   & h \in \U \\
			A_h\, s(\emptyset,\, \U),	& h \in \H \setminus \U.
		\end{cases}
		$
    \end{algorithmic}
\end{algorithm}

\begin{algorithm}[ht]
    \caption{{\em RNABOX}}
    \label{alg:rnabox}
    \textbf{Input:} $\H,\, (A_h)_{h \in \H},\, (m_h)_{h \in \H},\, (M_h)_{h \in \H},\, n$.
    \begin{algorithmic}[1]
		\Require $A_h > 0,\, 0 < m_h < M_h,\, h \in \H,\, \sum_{h \in \H} m_h \leq n \leq \sum_{h \in \H} M_h$.
		\State Set $\L = \emptyset$.
		\State\label{alg:rnabox:rna}Run {\em RNA}[$\H,\, (A_h)_{h \in \H},\, (M_h)_{h \in \H},\, n$] to obtain $(x^{**}_h,\, h \in 
		\H)$. \newline
		Let $\U = \{h \in \H:\, x^{**}_h = M_h\}$.
		\State\label{alg:rnabox:L}Determine $\Lt = \left\{h \in \H \setminus \U:\, x^{**}_h \leq m_h \right\}$. 
		\State\label{alg:rnabox:if}If {$\Lt =\emptyset$} go to {\footnotesize Step \ref{alg:rnabox:ret}}. Otherwise, update $n \gets 
		n - \sum_{h \in \Lt} m_h,\, \H \gets \H \setminus \Lt$, $\L \gets \L \cup \Lt$ and go to {\footnotesize Step 
		\ref{alg:rnabox:rna}}.
		\State\label{alg:rnabox:ret}Return $\x^* = (x^*_h,\, h \in \L \cup \H)$ with
		$x_h^* = 
		\begin{cases}
			m_h,		  & h \in \L \\
			x^{**}_h,	& h \in \H.
		\end{cases}
		$
    \end{algorithmic}
\end{algorithm}

We note that  in real life applications numbers $(A_h)_{h \in \H}$ are typically unknown and therefore their estimates $(\hat 
A_h)_{h \in \H}$ are used instead in the input of the algorithms.

Theorem \ref{th:rnabox} is the main theoretical result of this paper and its proof is given in Appendix \ref{app:rnabox_proof}.

\begin{theorem}
    \label{th:rnabox}
    The {\em RNABOX} algorithm provides the optimal solution to  Problem \ref{prob}.
\end{theorem}

\subsection{An example of performance of {\em RNABOX}} 
\label{sec:rnabox_example} 

We demonstrate the operational behaviour of {\em RNABOX} algorithm for an artificial population with 10 strata and for total 
sample size $n = 5110$, as shown in Table \ref{tab:rnabox_example}.
\begin{table}[ht]
    \caption{An example of {\em RNABOX} performance for a population with 10 strata and total sample size $n = 5110$. 
	Columns $\L_r/\U_r$, $r = 1, \ldots, 6$, represent the content of sets $\L,\, \U$ respectively, in the $r$-th iteration of the 
	{\em RNABOX} (between {\footnotesize Step \ref{alg:rnabox:L}} and {\footnotesize Step \ref{alg:rnabox:if}}): symbols 
	\tmin\hspace{0.1mm} or \tmax\hspace{0.1mm} indicate that the stratum with label $h$ is in $\L_r$ or $\U_r$, respectively.}
    \footnotesize
    \centering
    \label{tab:rnabox_example}
    \begin{tabular}[t]{| p{12mm} | r | r r || c c c c c c || S |}
        \hline
		$h$ & $A_h$ & $m_h$ & $M_h$ & $\L_1/\U_1$ & $\L_2/\U_2$ & $\L_3/\U_3$ & $\L_4/\U_4$ & $\L_5/\U_5$ & 
		$\L_6/\U_6$ & $\x^*$\\ 
		\hline
		1 	& 2700 & 750  & 900   & 		  & 		   & \tmin   & \tmin   & \tmin   & \tmin & 750\\
		2 	& 2000 & 450 & 500	 & \tmax & 		     & \tmin   & \tmin   & \tmin    & \tmin & 450\\
		3 	& 4200 & 250 & 300	 & \tmax & \tmax & \tmax  & \tmax  & \tmax  & 			& 261.08  \\
		4 	& 4400 & 350 & 400   & \tmax & \tmax & \tmax  & 		   &             & \tmin & 350 \\
		5 	& 3200 & 150  & 200   & \tmax & \tmax & \tmax  & \tmax  & \tmax	 & 			 & 198.92  \\
		6 	& 6000 & 550 & 600   & \tmax & \tmax & \tmax  & 		   & \tmin   & \tmin & 550 \\
		7 	& 8400 & 650 & 700    & \tmax & \tmax & \tmax & 		   & 			 & \tmin & 650 \\
		8	& 1900 & 50	   & 100    & \tmax & \tmax & \tmax & \tmax  & \tmax  & \tmax & 100 \\
		9 	& 5400 & 850 & 900   & \tmax & \tmax &  		  & \tmin  & \tmin    & \tmin & 850  \\
		10 & 2000 & 950 & 1000	&  		    & \tmin  & \tmin   & \tmin  & \tmin    & \tmin & 950\\
		\hline
		\multicolumn{2}{|l|}{SUM} & 5000 & 5600 & 0/8  & 1/7 & 3/6 & 4/3 & 5/3 & 7/1 &  5110 \\ 
		\hline
    \end{tabular}
\end{table}

For this example, {\em RNABOX} stops after $6$ iterations with {\em take-min} strata set $\L^* = \{1, 2, 4, 6, 7, 9, 10\}$, {\em 
take-max} strata set $\U^* = \{8\}$ and {\em take-Neyman} strata set $\H \setminus (\L^* \cup \U^*) = \{3, 5\}$ (see column 
$\L_6/\U_6$). The optimum allocation is a {\em regular} one and it is given in column $\x^*$ of Table 
\ref{tab:rnabox_example}. The corresponding value of the objective function \eqref{eq:prob_obj} is $441591.5$. The details
of interim allocations of strata to sets $\L,\, \U$ at each of $6$ iterations of the algorithm are given in columns $\L_1/\U_1$ -
$\L_6/\U_6$.

Results summarized in Table \ref{tab:rnabox_example} are presented graphically in Fig. \ref{fig:rnabox_choinka}.
\begin{figure}[ht]
    \caption{Assignments of strata into set $\L$ ({\em take-min}) and set $\U$ ({\em take-max}) in {\em RNABOX} algorithm for 
    an example of population as given in Table  \ref{tab:rnabox_example} and total sample size $n = 5110$. The $\sigma(\H)$ 
    axis corresponds to strata assigned to set $\U$, while $\tau(\H)$ is for $\L$. Squares represent assignments of strata to 
    $\L$ (\tmin) or $\U$ (\tmax) such that the coordinate corresponding to a given square is the value of the last element 
    (following the order of strata, $\sigma$ or $\tau$, associated to the respective axis) in the set.}
    \label{fig:rnabox_choinka}
    \begin{center}
        \includegraphics[scale=0.5]{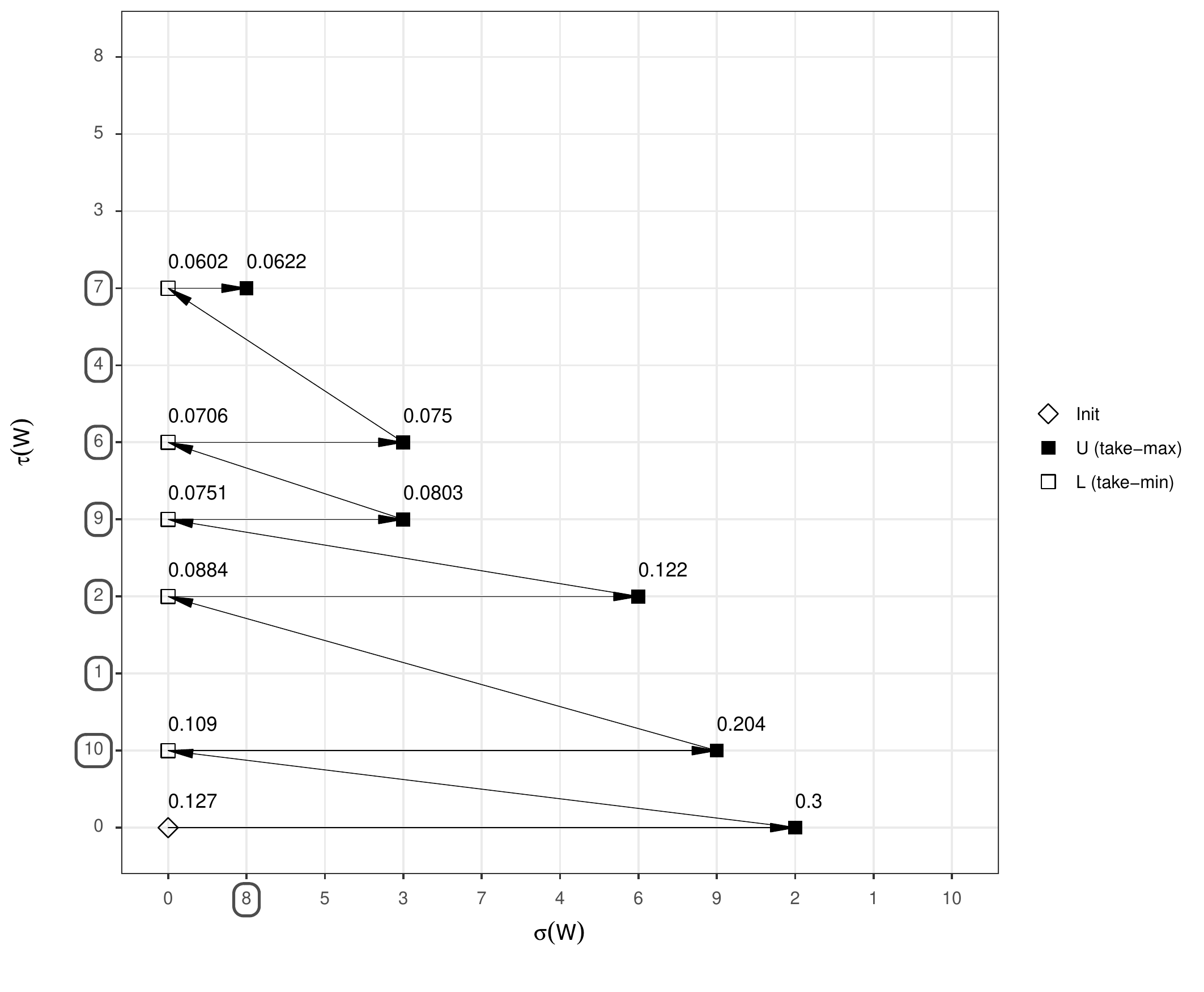}
    \end{center}
\end{figure}
The $\sigma(\H)$ axis corresponds to strata assigned to set $\U$, while $\tau(\H)$ is for $\L$. For the sake of clear 
presentation, the values on $\sigma(\H)$ axis are ordered non-increasingly in $\tfrac{A_h}{M_h}$, $h \in \H$, which gives 
$\tfrac{A_8}{M_8} \ge \tfrac{A_5}{M_5} \ge \tfrac{A_3}{M_3} \ge \ldots \ge \tfrac{A_{10}}{M_{10}}$, while for $\tau(\H)$ axis 
the strata are ordered non-decreasingly in $\tfrac{A_h}{m_h}$, $h \in \H$, which gives $\tfrac{A_{10}}{m_{10}} \le 
\tfrac{A_1}{m_1} \le \tfrac{A_2}{m_2} \le \ldots \le \tfrac{A_8}{m_8}$. Squares in Fig. \ref{fig:rnabox_choinka} represent 
assignments of strata into sets $\L$ or $\U$, such that the corresponding coordinate is the value of the last element (follow 
the respective axis) in the set. For example, \tmax\, at coordinate $6$ on $\sigma(\H)$ axis means that in this  iteration of the 
algorithm (iteration no. $3$) we have $\U = \{8, 5, 3, 7, 4, 6\}$, while \tmin\, at coordinate $6$ on $\tau(\H)$ axis means that 
in this iteration (no. $5$) we get $\L = \{10, 1, 2, 9, 6\}$. Directions indicated by the arrows correspond to the order in which 
the sets $\L$ and $\U$ are built as the algorithm iterates. Numbers above the squares are the values of set function $s(\L,\, 
\U)$ for  $\L$ and $\U$ with elements indicated by the coordinates of a given square. For example, for \tmax\, at coordinates 
$(8,\, 7)$, we have $s(\{10, 1, 2, 9, 6, 4, 7\},\, \{8\}) = 0.0622$.

We would like to point out that regardless of the population and other allocation parameters chosen, the graph illustrating the 
operations of {\em RNABOX} algorithm will always have a shape similar to that of the right half of a Christmas tree with the 
top cut off when $\U_{r^*} \neq \emptyset$. This property results from that fact that $\L_r \subsetneq \L_{r+1}$(see 
\eqref{eq:rnabox_Lsum}) and $\U_r \supseteq \U_{r+1}$ (see \eqref{eq:lem:U}) for $r = 1, \ldots, r^* - 1,\, r^* \geq 2$. For the 
population given in Table \ref{tab:rnabox_example}, we clearly see that $\L_1 = \emptyset \subsetneq \L_2 = \{10\} \subsetneq 
\L_3 = \{10, 1, 2\} \subsetneq \L_4 = \{10, 1, 2, 9\} \subsetneq \L_5 = \{10, 1, 2, 9, 6\} \subsetneq \L_6 = \L^* = \{10, 1, 2, 9, 6, 4, 
7\}$ and $\U_1 = \{8, 5, 3, 7, 4, 6, 9, 2\} \supset \U_2 = \{8, 5, 3, 7, 4, 6, 9\} \supset \U_3 = \{8, 5, 3, 7, 4, 6\} \supset \U_4 = 
\{8, 5, 3\}  = \U_5 \supset \U_6 = \U^* = \{8\}$. Moreover, regardless of the population and other allocation parameters 
chosen, subsequent values of set function $s$, which are placed at the ends of branches of the Christmas tree, above \tmax, 
form a non-increasing sequence while moving upwards. This fact follows directly from Lemma \ref{lem:s_in_rnabox}. For the 
example allocation illustrated in Fig. \ref{fig:rnabox_choinka}, we have $s(\L_1,\, \U_1) = 0.3 > s(\L_2,\, \U_2) = 0.204 > 
s(\L_3,\, \U_3) = 0.122 > s(\L_4,\, \U_4) = 0.0803 > s(\L_5,\, \U_5) = 0.075 > s(\L_6,\, \U_6) = 0.0622$ (note that $\L_1 = 
\emptyset$). The same property appears for values of $s$ related to the trunk of the Christmas tree, placed above \tmin. In 
this case, this property is due to \eqref{eq:lem:smono:1}, \eqref{eq:rem:rnabox_UUL:L} and \eqref{eq:rnabox_Lsum}. For 
allocation in Fig. \ref{fig:rnabox_choinka}, we have $s(\L_1,\, \emptyset) = 0.127 > s(\L_2,\, \emptyset) = 0.109 > s(\L_3,\, 
\emptyset) = 0.0884 > s(\L_4,\, \emptyset) = 0.0751 > s(\L_5,\, \emptyset) = 0.0706 > s(\L_6,\, \emptyset) = 0.0602$.

\subsection{Possible modifications and improvements} 
\label{ssec:rnabox_improv}

\subsubsection{Alternatives for {\em RNA} in Step 2} 

The {\em RNABOX} algorithm uses {\em RNA} in its {\footnotesize Step \ref{alg:rnabox:rna}}. However, it is not hard to see 
that any algorithm dedicated to Problem \ref{prob:upper} (like for instance {\em SGA} by \citealp{SG} or {\em COMA} by 
\citealp{WWW}) could be used instead. We chose  {\em RNA} as it allows to keep {\em RNABOX} free of any strata sorting.

\subsubsection{A twin version of {\em RNABOX}} 

Let us observe that the order in which $\L$ and $\U$ are computed in the algorithm could be interchanged. Such a change, 
implies that the {\em RNA} used in {\footnotesize Step \ref{alg:rnabox:rna}} of the {\em RNABOX}, should be replaced by its 
twin version, the {\em LRNA}, that solves optimum allocation problem under one-sided lower bounds. The {\em LRNA} is 
described in details in \cite{WojciakLRNA}. 

\begin{algorithm}[ht]
	\caption{{\em LRNA}}
	\textbf{Input:} $\H,\, (A_h)_{h \in \H},\, (m_h)_{h \in \H},\, n$.
	\begin{algorithmic}[1]
		\Require $A_h > 0,\, m_h > 0,\, h \in \H$, $n \geq \sum_{h \in \H}\, m_h$.
		\State Let $\L = \emptyset$.
		\State\label{alg:lrna:step_R}Determine $\Lt = \left\{h \in \H \setminus \L:\, A_h\, s(\L,\, \emptyset) \leq m_h \right\}$,
		where set function $s$ is defined in \eqref{eq:s}.
		\State\label{alg:lrna:step_check} If {$\Lt = \emptyset$}, go to \footnotesize Step \ref{alg:lrna:step_return} \normalsize. 
		Otherwise, update $\L \gets \L \cup \Lt$  and go to \footnotesize Step \ref{alg:lrna:step_R} \normalsize.
		\State\label{alg:lrna:step_return}Return $\x^* = (x^*_h,\, h \in \H)$ with
		$x^*_h =
		\begin{cases}
			m_h,	& h \in \L \\
			A_h\, s(\L,\, \emptyset),	& h \in \H \setminus \L.
		\end{cases}
		$
	\end{algorithmic}
\end{algorithm}

Taking into account the observation above, {\footnotesize Step \ref{alg:rnabox:rna}} and {\footnotesize Step 
\ref{alg:rnabox:L}} of {\em RNABOX} would read: \\
\indent {\footnotesize Step 2:} ~Run {\em LRNA}[$\H,\, (A_h)_{h \in \H},\, (m_h)_{h \in \H},\, n$] to obtain $(x^{**}_h,\, h \in 
\H)$. \\
\indent \indent Let $\L = \{h \in \H:\, x^{**}_h = m_h\}$. \\
\indent  {\footnotesize Step 3:} ~Determine $\Ut = \left\{h \in \H \setminus \L:\, x^{**}_h \geq M_h \right\}$. \\
The remaining steps should be adjusted accordingly.

\subsubsection{Using prior information in {\em RNA} at Step \ref{alg:rnabox:rna}} 

In view of Lemma \ref{lem:U}, using the notation introduced in Appendix \ref{app:notation}, in {\footnotesize Step 
\ref{alg:rnabox:rna}} of {\em RNABOX}, for $r^* \geq 2$ we have
\begin{equation*}
	\U_r = \{h \in \H \setminus \L_r:\, x^{**}_h = M_h\} \subseteq \U_{r-1}, \ind{r = 2, \ldots, r^*.}
\end{equation*}
This suggests that the domain of discourse for $\U_r$ could be shrunk from $\H \setminus \L_r$ to $\U_{r-1} \subseteq \H 
\setminus \L_r$, i.e.
\begin{equation}
	\U_r = \{h \in \U_{r-1}:\, x^{**}_h = M_h\}, \ind{r = 2, \ldots, r^*.}
\end{equation}
Given the above observation and the fact that from the implementation point of view set $\U_r$ is determined internally by 
{\em RNA}, it is tempting to consider modification of {\em RNA} such that it makes use of the domain of discourse 
$\U_{r-1}$ for set $\U_r$. This domain could be specified as an additional input parameter, say $\mathcal J \subseteq \H$, 
and then {\footnotesize Step \ref{alg:rna:U}} of {\em RNA} algorithm would read: \\
\indent {\footnotesize Step 2:} ~Determine $\Ut = \left\{h \in \mathcal J \setminus \U:\, A_h \, s(\emptyset,\, \U) \ge 
M_h\right\}$. \\
From {\em RNABOX} perspective, this new input parameter of {\em RNA} should be set to $\mathcal J = \H$ for the first 
iteration, and then $\mathcal J = \U_{r-1}$ for subsequent iterations $r = 2, \ldots, r^* \geq 2$ (if any).

\subsection{On a naive extension of the one-sided {\em RNA}} 
\label{sec:naive_rnabox}
By the analogy to the one-sided constraint case, one would expect that the optimum allocation problem under two-sided 
constraints \eqref{eq:prob_cnst_ineq} is of the form \eqref{eq:LU} and hence it could be solved with the naive modification of 
{\em RNA} as defined below.
\begin{algorithm}[H]
    \caption{Naive modification of {\em RNA}}
    \label{alg:naive_mod_RNA}
    \textbf{Input:} $\H,\, (A_h)_{h \in \H},\, (m_h)_{h \in \H},\, (M_h)_{h \in \H},\, n$.
    \begin{algorithmic}[1]
		\Require $A_h > 0,\, 0 < m_h < M_h,\, h \in \H,\, \sum_{h \in \H} m_h \leq n \leq \sum_{h \in \H} M_h$.
		\State\label{alg:naive_rnabox:LU_init}Set $\L = \emptyset,\, \U = \emptyset$.
		\State\label{alg:naive_rnabox:U}Determine $\Ut = \{h \in \H \setminus (\L \cup \U):\, A_h\, s(\L,\, \U) \geq M_h\}$, where 
		set function $s$ is defined in \eqref{eq:s}.
		\State\label{alg:naive_rnabox:L}Determine $\Lt = \{h \in \H \setminus (\L \cup \U):\, A_h\, s(\L,\, \U) \leq m_h\}$.
		\State\label{alg:naive_rnabox:if}If {$\Lt \cup \Ut =\emptyset$} go to {\footnotesize Step \ref{alg:naive_rnabox:ret}}. 
		Otherwise, update $\L \gets \L \cup \Lt,\, \U \gets \U \cup \Ut$ and go to {\footnotesize Step \ref{alg:naive_rnabox:U}}.
		\State\label{alg:naive_rnabox:ret}Return $\x^* = (x^*_h,\, h \in \H)$ with
		$
		x^*_h = 
		\begin{cases}
			m_h,		  			& ~ h \in \L \\
			M_h,		  			& ~h \in \U \\
			A_h\, s(\L,\, \U),	& ~ h \in \H \setminus (\L \cup \U).
		\end{cases}
		$
    \end{algorithmic}
\end{algorithm}
This is however not true and the basic reason behind is that sequence $\left(s(\L_r,\, \U_r)\right)_{r = 1}^{r^*}$ (where $r$ 
denotes iteration index, and $\L_r$ and $\U_r$ are taken after {\footnotesize Step \ref{alg:naive_rnabox:L}} and before 
{\footnotesize Step \ref{alg:naive_rnabox:if}}), typically fluctuates. Thus, it may happen that \eqref{eq:rem:rnabox_UUL:U} - 
\eqref{eq:rem:rnabox_UUL:L} are not met for $\U_r,\, \Lt_r$ defined by the above naive algorithm for some $r = 1, \ldots, r ^* 
\geq 1$. Consequently, optimality conditions stated in Theorem \ref{th:optcond} may not hold. To illustrate this fact, consider 
an example of population as given in Table \ref{tab:rna_naive_example}.
\begin{table}[ht]
    \caption{Example population with 5 strata and the results of naive recursive Neyman sample allocation (columns $\x$ and 
    $\L/\U/\H \setminus (\L \cup \U)$) for total sample size $n = 1489$. Lower and upper bounds imposed on strata sample 
    sizes are given in columns $m_h$ and $M_h$ respectively. The optimum allocation with corresponding strata sets are 
    given in the last two columns. Squares in cells indicate that the allocation for stratum $h \in \H = \{1,\, \ldots,\ 5\}$ is of type 
    {\em take-min} (\tmin) or {\em take-max} (\tmax).}
    \footnotesize
    \centering
    \label{tab:rna_naive_example}
    \begin{tabular}[t]{ p{12mm} r r r || S c | S c}
		\hline
		$h$ & $A_h$ & $m_h$ & $M_h$ & $\x$ & $\L/\U/\H \setminus (\L \cup \U)$ & $\x^*$ & $\L^*/\U^*/\H \setminus (\L^* 
		\cup \U^*)$ \\
		\hline
		1 & 420     & 24    & 420    & 30     &  		  &  54.44 &  \\
		2 & 352    & 15     & 88      & 88     & \tmax & 45.63  & \\
		3 & 2689 & 1344 & 2689 & 1344 & \tmin  & 1344    & \tmin \\
		4 & 308   & 8       & 308    & 22     &  		   &  39.93 &  \\
		5 & 130    & 3       & 5         & 5       & \tmax & 5          & \tmax  \\ 
		\hline
		SUM &  & 1394 & 3510 & 1489 & 1/2/2 & 1489 & 1/1/3 \\ 
		\hline
    \end{tabular}
\end{table}
For this population, a naive version of the recursive Neyman stops in the second iteration and it gives allocation $\x$ with $\L 
= \{3\}$ and $\U = \{2, 5\}$, while the optimum choice is $\x^*$ with $\L^* = \{3\}$ and $\U^* = \{5\}$. For non-optimum 
allocation $\x$, we have $s(\{3\},\, \{2, 5\}) = 0.0714 \ngeq 0.25 = \tfrac{A_2}{M_2}$, i.e. $\U$ does not follow 
\eqref{eq:optcond}. Note that the allocation $\x$ is feasible in this case (it may not be in general), yet is not a minimizer of the 
objective function \eqref{eq:prob_obj} as $f(\x) = 20360 - B$, and $f(\x^*) = 17091 - B$, where $B$ is some constant. Thus, 
it appears that in the case of two-sided constraints, one needs a more subtle modification of {\em RNA} procedure that lies in 
a proper update of the pair $(\L_r,\, \U_r)$ in each iteration $r = 1, \ldots, r^* \geq 1$ of the algorithm, such as e.g. in 
{\em RNABOX}.

\section{Numerical results} 
\label{sec:numerical_results}

In simulations, using \proglang{R} Statistical Software \citep{R} and {\tt microbenchmark} \proglang{R} package 
\citep{microbenchmark}, we compared the computational efficiency of {\em RNABOX} algorithm with the efficiency of the 
{\em fixed-point iteration} algorithm ({\em FPIA}) of \citet{MSW}. The latter one is known to be an efficient algorithm 
dedicated to Problem \ref{prob} and therefore we used it as a benchmark. The comparison was not intended to verify 
theoretical computational complexity, but was rather concerned with quantitative results regarding computational efficiency 
for the specific implementations of both algorithms.

To compare the performance of the algorithms we used the {\em STSI} sampling for several artificially created populations. 
Here, we chose to report on simulation for two such populations with 691 and 703 strata, results of which are representative 
for the remaining ones. These two populations were constructed by iteratively binding $K$ sets of numbers, where $K$ 
equals 100 (for the first population) and 200 (for the second population). Each set, labelled by $i = 1, \dots, K$, contains 
10000 random numbers generated independently from log-normal distribution with parameters $\mu = 0$ and $\sigma = \log 
(1+i)$. For every set $i = 1, \dots, K$, strata boundaries were determined by the geometric stratification method of \citet{GH} 
with parameter 10 being the number of  strata and targeted coefficient of variation equal to 0.05. This stratification method is 
implemented in the \proglang{R} package {\tt stratification}, developed by \citet{stratification} and described in 
\citet{Baillargeon}. For more details, see the \proglang{R} code with the experiments, which is placed in our GitHub repository 
\citep[see][]{rnabox_perf}.

Results of these simulations are illustrated in Fig. \ref{fig_time_comparisons}. 
\begin{figure}[ht]
    \caption{Running times of {\em FPIA} and {\em RNABOX} for two artificial populations. Top graphs show the empirical 
    median of execution times (calculated from 100 repetitions) for different total sample sizes. Numbers in brackets are the 
    numbers of iterations of a given algorithm. In the case of {\em RNABOX}, it is a vector with number of iterations of the {\em 
    RNA} (see {\footnotesize Step \ref{alg:rnabox:rna}} of {\em RNABOX}) for each iteration of {\em RNABOX}. Thus, 
    the length of this vector is equal to the number of iterations of {\em RNABOX}. Counts of {\em take-min}, {\em 
    take-Neyman}, and {\em take-max} strata are shown on bottom graphs.}
    \label{fig_time_comparisons} 
    \begin{center}
		\includegraphics[scale=0.54]{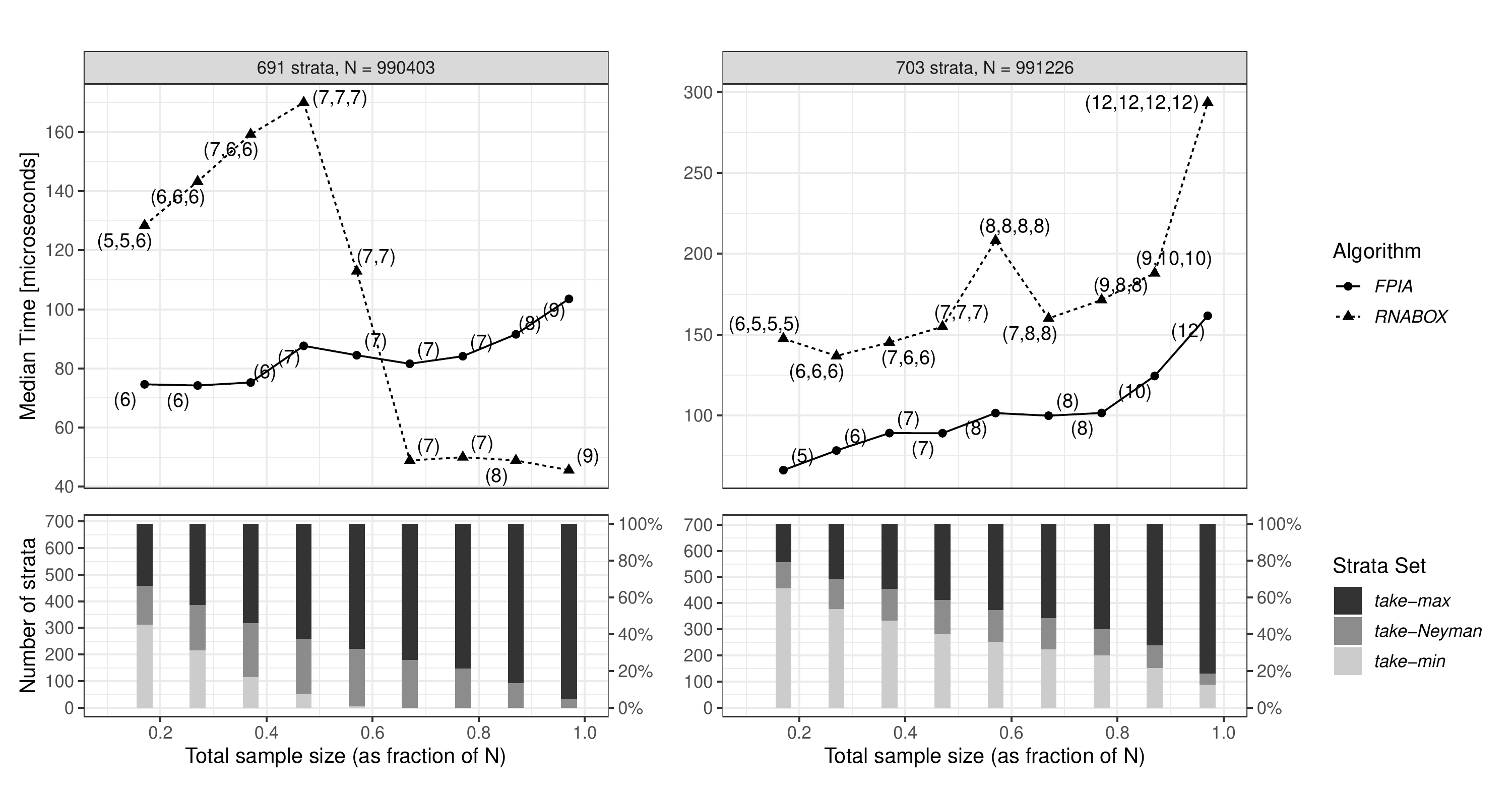}
    \end{center}
\end{figure}
From Fig. \ref{fig_time_comparisons} we see that, while for majority of the cases the {\em FPIA} is slightly faster than {\em 
RNABOX}, the running times of both of these algorithms are generally comparable. The gain in the execution time of the 
{\em FPIA} results from the fact that it typically runs through a smaller number of sets $\L,\, \U \subset \H$ (according to 
\eqref{eq:fpia_crucialstep}), than RNABOX in order to find the optimal $\L^*$ and $\U^*$. Although this approach usually 
gives correct results (as in the simulations reported in this section), it may happen that the {\em FPIA} misses the optimal sets 
$\L^*,\, \U^* \subsetneq \H$. In such a case, {\em FPIA} does not converge, as explained in Appendix \ref{app:msw}. 
Nevertheless, we emphasize that this situation rarely happens in practice.

\section{Concluding comments} 
\label{sec:conclusions}

In this paper we considered Problem \ref{prob} of optimum sample allocation under box constraints. The main result of this 
work is the mathematically precise formulation of necessary and sufficient conditions for the solution to Problem 
\ref{prob}, given in Theorem \ref{th:optcond}, as well as the development of the new recursive algorithm, termed {\em 
RNABOX}, that solves Problem \ref{prob}. The optimality conditions are fundamental to analysis of the optimization problem. 
They constitute trustworthy underlay for  development of effective algorithms and can be used as a baseline for any future 
search of new algorithms solving Problem \ref{prob}. Essential properties of {\em RNABOX} algorithm, that distinguish it 
from other existing algorithms and approaches to the Problem \ref{prob}, are:

\begin{enumerate}[nosep, labelsep=5pt]
    \item Universality: {\em RNABOX} provides optimal solution to every instance of feasible Problem \ref{prob} (including the case of a {\em vertex} optimum allocation).
    \item No initialization issues: {\em RNABOX} does not require any initializations, pre-tests or whatsoever that could have an 
    impact on the final results of the algorithm. This, in turn, takes places e.g. in case of NLP methods.
    \item No sorting: {\em RNABOX} does not perform any ordering of strata.
    \item Computational efficiency: {\em RNABOX} running time is comparable to that of {\em FPIA} (which is probably the 
    fastest previously known optimum allocation algorithm for the problem considered).
    \item Directness: {\em RNABOX} computes important quantities (including {\em RNA} internals) via formulas that are 
    expressed directly in terms of the allocation vector $\x^{(\L,\, \U)}$ (see Definition \ref{def:LU}). This 
    reduces the risk of finite precision arithmetic issues, comparing to the algorithms that base their key operations on some 
    interim variables on which the optimum allocation depends, as is the case of e.g. the NLP-based method.
    \item Recursive nature: {\em RNABOX} repeatedly applies allocation {\footnotesize Step \ref{alg:rnabox:rna}} and 
    {\footnotesize Step \ref{alg:rnabox:L}} to step-wise reduced set of strata, i.e. "smaller" versions of the same problem. This 
    translates to clarity of the routines and a natural way of thinking about the allocation problem.
    \item Generalization: {\em RNABOX}, from the perspective of its construction, is a generalization of the popular {\em RNA} 
    algorithm that solves Problem \ref{prob:upper} of optimum sample allocation under one-sided bounds on sample strata 
    sizes.
\end{enumerate}

\bigskip
Finally, we would like to note that Problem \ref{prob} considered in this paper is not an integer-valued allocation problem, 
while the sample sizes in strata should be of course natural numbers. On the other hand, the integer-valued optimum 
allocation algorithms are relatively slow and hence might be inefficient in some applications, as already noted in Section 
\ref{sec:motivation}. If the speed of an algorithm is of concern and non-necessarily integer-valued allocation algorithm is 
chosen (e.g. {\em RNABOX}), the natural remedy is to round the non-integer optimum allocation provided by that algorithm. 
Altogether, such procedure is  still much faster than integer-valued allocation algorithms. However, a simple rounding of the 
non-integer solution does not, in general, yield the minimum of the objective function, and may even lead to an infeasible 
solution, as noted in \citet[Section 1, p.~3]{Friedrich}. Since infeasibility can in fact arise only from violating constraint 
\eqref{eq:prob_cnst_eq}, it can be easily avoided by using a rounding method of \citet{Cont} that preserves the integer sum 
of positive numbers. Moreover, all numerical experiments that we carried out, show that the values of the objective function 
obtained for non-integer optimum allocation before and after rounding and for the integer optimum allocation are practically 
indistinguishable. See Table \ref{tab:varround} for the exact numbers. 

\begin{table}[ht]
    \caption{For populations used in simulations in Section \ref{sec:numerical_results}, $V_{int}$ and $V$ denote values of 
    variance \eqref{eq:var} computed for integer optimum allocations and non-integer optimum allocations, respectively. 
    Variances $V_{round}$ are computed for rounded non-integer optimum allocations (with the rounding method of 
    \citet{Cont}).}
    \footnotesize
    \centering
    \label{tab:varround}
    \begin{tabular}[t]{ | r | r | r | r | r | r | r |}
		\hline
		& \multicolumn{3}{c|}{$691 \text{ strata } ~ (N = 990403)$} & \multicolumn{3}{c|}{$703 \text{ strata } ~ (N = 991226)$} \\
		\hline
		fraction f & $n = \text{ f } \cdot N$ & $V/V_{int}$ & $V_{round}/V_{int}$ & $n = \text{ f } \cdot N$ & 
		$V/V_{int}$ & $V_{round}/V_{int}$ \\ \hline
		0.1 & 99040 & 0.999997 & 1.000000 & 99123 & 0.999997 & 1.00000 \\ \hline
		0.2 & 198081 & 0.999999 & 1.000000 & 198245 & 0.999999 & 1.00000 \\ \hline
		0.3 & 297121 & 0.999999 & 1.000000 & 297368 & 0.999999 & 1.00000 \\ \hline
		0.4 & 396161 & 0.999999 & 1.000000 & 396490 & 0.999999 & 1.00000 \\ \hline
		0.5 & 495202 & 0.999999 & 1.000000 & 495613 & 0.999999 & 1.00000 \\ \hline
		0.6 & 594242 & 0.999999 & 1.000000 & 594736 & 0.999999 & 1.00000 \\ \hline
		0.7 & 693282 & 1.000000 & 1.000000 & 693858 & 0.999998 & 1.00000 \\ \hline
		0.8 & 792322 & 1.000000 & 1.000000 & 792981 & 0.999995 & 1.00000 \\ \hline
		0.9 & 891363 & 1.000000 & 1.000000 & 892103 & 0.999759 & 1.00000 \\ \hline
    \end{tabular}
\end{table}

The above observations suggest that fast, not-necessarily integer-valued allocation algorithms, with properly rounded 
results, may be a good and reasonable alternative to slower integer algorithms when speed of an algorithm is crucial.

\appendix

\section{Appendix: On the existing allocation algorithms by GGM and MSW}
\label{app:ggm_msw}

\subsection{Necessary vs. sufficient optimality conditions and the {\em noptcond} function of GGM} 
\label{app:ggm}

Problem \ref{prob} has been considered in \citet{GGM} (GGM in the sequel). Theorem 1 in that paper announces that there 
exist disjoint sets $\L,\, \U \subset \H$ with $\L^c = \H \setminus \L,\, \U^c = \H \setminus \U$, such that
\begin{equation}
    \label{eq:ggm:noptcond_LU}
    \max_{h \in \L} \tfrac{A_h}{m_h} < \min_{h \in \L^c} \tfrac{A_h}{m_h} 
    \qquad\mbox{and}\qquad
    \max_{h \in \U^c} \tfrac{A_h}{M_h} < \min_{h \in \U} \tfrac{A_h}{M_h},
\end{equation}
and, if the allocation $\x = (x_h,\, h \in \H)$ is of the form \eqref{eq:LU} with $\L$ and $\U$ as above, then $\x$ is an optimum 
allocation. An explanation on how to construct sets $\L$ and $\U$ is essentially given in the first part of the proof of Theorem 
1, p. 154-155, or it can be read out from the {\em noptcond} function. The {\em noptcond} is a function in \proglang{R} 
language \citep[see][]{R} that was defined in Sec. 3 of GGM and it aims to find an optimum sample allocation based on results 
of Theorem 1, as the authors pointed out. However, as it will be explained here, \eqref{eq:LU} and 
\eqref{eq:ggm:noptcond_LU} together, are in fact necessary but not sufficient conditions for optimality of the allocation $\x$. 
In a consequence, the \emph{noptcond} function may not give an optimum allocation.

By Theorem \ref{th:optcond}, CASE I, a {\em regular} allocation $\x$ is an optimum allocation if and only if it is of the form \eqref{eq:LU} with disjoint sets $\L,\, \U \subsetneq \H$, $\L \cup \U \subsetneq \H$, $\L^c = \H \setminus \L,\, \U^c = \H 
\setminus \U$ such that the following inequalities hold true:
\begin{gather}
    \max_{h \in \L} \tfrac{A_h}{m_h} \leq \tfrac{1}{s(\L,\, \U)} < \min_{h \in \L^c} \tfrac{A_h}{m_h}, 
    \label{eq:ggm:true_optcond_L} \\
    \max_{h \in \U^c} \tfrac{A_h}{M_h} < \tfrac{1}{s(\L,\, \U)} \leq \min_{h \in \U} \tfrac{A_h}{M_h}. 
    \label{eq:ggm:true_optcond_U}
\end{gather}
Clearly, \eqref{eq:ggm:true_optcond_L} implies the first inequality in \eqref{eq:ggm:noptcond_LU}, and \eqref{eq:ggm:true_optcond_U} implies the second one. Converse implications do not necessary hold, and therefore sufficiency of condition \eqref{eq:LU} with \eqref{eq:ggm:noptcond_LU} for Problem \ref{prob} is not guaranteed. In others words, \eqref{eq:LU} and \eqref{eq:ggm:noptcond_LU} alone, lead to a feasible solution (i.e. the one that does not violate any 
of the constraints \eqref{eq:prob_cnst_eq} - \eqref{eq:prob_cnst_ineq}), which, at the same time, might not be a minimizer of the objective function \eqref{eq:prob_obj}.

The proof of Theorem 1 in GGM gives an explicit, but somewhat informal algorithm that finds sets $\L$ and $\U$ which define allocation \eqref{eq:LU}. Sets $\L$ and $\U$ determined by this algorithm meet inequalities (4) and (5) in GGM, which are 
necessary and sufficient conditions for an optimal ({\em regular}) solution. This fact was noted by the authors, but unfortunately (4) and (5) are not given in the formulation of Theorem 1. Presumably, this was due to authors' statement that 
(4) and (5) follow from \eqref{eq:LU} and \eqref{eq:ggm:noptcond_LU}: {\it "From the definition of L1, L2 in Eq. 3 we have ..."}, while in fact to conclude (4) and (5) the authors additionally rely on the algorithm constructed at the beginning of the proof. Note that conditions (4) and (5) are the same as those given in Theorem \ref{th:optcond} for a {\em regular} solution (CASE I). Moreover, it is a matter of simple observation to see that the allocation found by the algorithm described in the proof of Theorem 1 in GGM, meets also optimality conditions \eqref{eq:optcond_vertex} and \eqref{eq:optcond_vertex_n} established in Theorem \ref{th:optcond} for {\em vertex} solutions (CASE II). Hence, the allocation  computed by the algorithm embedded in the proof of the GGM's Theorem 1 is indeed an optimal one, though the formal statement of Theorem 1 is not correct. 

Unfortunately, the  {\em noptcond} function does not fully follow the algorithm from the proof of Theorem 1 of GGM. Consequently, the optimality of a solution computed by {\em noptcond} is not guaranteed. This can be illustrated by a simple 
numerical Example \ref{example:noptcond_wrong}, which follows \citet[Example 3.9]{WojciakMSC}.

\begin{example}
    \label{example:noptcond_wrong}
    Consider the allocation for an example population as given in Table \ref{tab:ggm}.
    \begin{table}[H]
        \caption{Two allocations for an example population with two strata: non-optimum $\x^{noptcond}$ with $\L = \{1\},\, \U = 
        \emptyset$, and optimum $\x^*$ with $\L^* = \emptyset,\, \U^* = \{1\}$. Set function $s$ is defined in \eqref{eq:s}.}
	\footnotesize
	\centering
	\label{tab:ggm}
 	\begin{tabular}[t]{| l | r | r r | S S || c c | S S |}
		\hline
		$h$ & $A_h$ & $m_h$ & $M_h$ & $\tfrac{A_h}{m_h}$ & $\tfrac{A_h}{M_h}$ & $s^{-1}(\{1\}, \emptyset)$ & $s^{-1}(\emptyset, \{1\})$ & $\x^{noptcond}$ & $\x^*$ \\
		\hline
		1 & 2000 & 30 & 50 & 66.67 & 40 & \multirow{2}{*}{23.08} & \multirow{2}{*}{27.27} & 30 & 50\\
		2 & 3000 & 40 & 200 & 75 & 15 & & & 130 & 110 \\ 
		\hline
	\end{tabular}
    \end{table}	
    The solution returned by {\em noptcond} function is equal to $\x = (30,\, 130)$, while the optimum allocation is $\x^* = (50, 
    110)$. The reason for this is that conditions \eqref{eq:ggm:true_optcond_L}, \eqref{eq:ggm:true_optcond_U} are never 
    examined by the {\em noptcond}, and for this particular example we have $\tfrac{A_1}{m_1} = 66.67 \nleq 23.08 = 
    \tfrac{1}{s(\L,\, \U)}$ as well as $\tfrac{A_1}{M_1} = 40 \nless 23.08 = \tfrac{1}{s(\L,\, \U)}$, i.e. 
    \eqref{eq:ggm:true_optcond_L}, \eqref{eq:ggm:true_optcond_U} are clearly not met. A simple adjustment can be made to 
    {\em noptcond} function so that it provides the optimal solution to Problem \ref{prob}. That is, a feasible candidate solution 
    that is found (note that this candidate is of the form $\x = (x_v,\, v \in \H \setminus (\L \cup \U))$ with $x_v = A_v\, s(\L,\, 
    \U)$), should additionally be checked against the condition
    \begin{equation}
		\max_{h \in \L} \tfrac{A_h}{m_h} \leq \tfrac{A_v}{x_v} \leq \min_{h \in \U} \tfrac{A_h}{M_h}, \ind{v \in \H \setminus (\L \cup 
		\U).}
    \end{equation}
\end{example} 

We finally note that the table given in Section 4 of GGM on p. 160, that illustrates how {\em noptcond} operates, is incorrect as some of the numbers given do not correspond to what {\em noptcond} internally computes, e.g. for each row of this table, the allocation should sum up to $n = 20$, but it does not.

\subsection{Fixed-point iteration of MSW} 
\label{app:msw}

Theorem 1 of \citet{MSW} (referred to by MSW in the sequel) provides another solution to Problem \ref{prob} that is based on 
results partially similar to those stated in Theorem \ref{th:optcond}. Specifically, it states that the allocation vector $\x = 
(x_h,\, h \in \H)$ can be expressed as a function $\x:\, \R_+ \to \R_+^{\card{\H}}$ of $\lambda$, defined as follows
\begin{equation*}
    x_h(\lambda) =
    \begin{cases}
		M_h,										&\qquad{h \in J^{\lambda}_M := \left\{h \in \H:\, \lambda \le \tfrac{A_h^2}{M_h^2}\right\}} \\
		\tfrac{A_h}{\sqrt{\lambda}}, &\qquad{h \in J^{\lambda} := \left\{h \in \H:\, \tfrac{A_h^2}{M_h^2} < \lambda < 
		\tfrac{A_h^2}{m_h^2}\right\}} \\
		m_h,										&\qquad{h \in J^{\lambda}_m := \left\{h \in \H:\, \lambda \ge \tfrac{A_h^2}{m_h^2} \right\},}
    \end{cases}
\end{equation*}	
and the optimum allocation $\x^*$ is obtained for $\lambda = \lambda^*$, where $\lambda^*$ is the solution of the equation:
\begin{equation}
    \label{eq:msw:g}
    {\tilde{g}}(\lambda) := \sum_{h \in \H}\, x_h(\lambda) - n  = 0.
\end{equation}
Here, we note that the optimum allocation $\x^*$ obtained is of the form $\x^{(\L,\, \U)}$, as given in \eqref{eq:LU}, with $\L = 
J^{\lambda^*}_m$ and $\U = J^{\lambda^*}_M$. Since $\tilde{g}$ is continuous but not differentiable, equation 
\eqref{eq:msw:g} has to be solved by a root-finding methods which only require continuity. In MSW, the following methods were proposed: {\em bisection}, {\em secan} and {\em regula falsi} in few different versions. The authors reported 
that despite the fact that some theoretical requirements for these methods are not satisfied by the function $\tilde{g}$ (e.g. it 
is not convex while {\em regula falsi} applies to convex functions), typically the numerical results are satisfactory. However 
this is only true when the values of the initial parameters are from the proper range, which is not known a priori. 

As these algorithms might be relatively slow (see Table 1 in MSW), the authors of MSW proposed an alternative method. It 
is based on the observation that \eqref{eq:msw:g} is equivalent to
\begin{equation*}
    \lambda = \phi(\lambda) := \tfrac{1}{s^2(J_m^{\lambda},\,J^{\lambda}_M)},
\end{equation*}
where set function $s$ is as in \eqref{eq:s} (see Appendix A, Lemma 2 in MSW). Note that $\phi$ is well-defined only if 
$J_m^{\lambda} \cup J_M^{\lambda} \subsetneq \H$. The authors note that the optimal $\lambda^*$ is a fixed-point of 
the function $\phi: \R_+ \to \R_+$. This clever observation was translated to an efficient optimum allocation {\em fixed-point 
iteration} algorithm, {\em FPIA}, that was defined on page 442 in MSW. Having acknowledged its 
computational efficiency as well as the fact that typically it gives the correct optimum allocation, the following two minor 
issues related to {\em FPIA} can be noted:

\begin{itemize}[nosep, labelsep=5pt]
    \item The {\em FPIA} is adequate only for allocation problems for which an optimum allocation is of a {\em regular} type 
    (according to Definition \ref{def:vertex_regular}), as for {\em vertex} allocation $J_m^{\lambda^*} \cup J_M^{\lambda^*} = 
    \H$ and therefore $\phi$ is not well-defined for such $\lambda^*$. 
    \item The {\em FPIA} strongly depends on the choice of initial value $\lambda_0$ of the parameter $\lambda$. If 
    $\lambda_0$ is not chosen from the proper range (not known a priori), two distinct undesirable scenarios can occur: the 
    {\em FPIA} may get blocked, or it may not converge. 
\end{itemize}

To outline the blocking scenario which may happen even in case of a {\em regular} optimum allocation, let
\begin{equation*}
	I := \left\{\lambda \in \left[\min_{h \in \H} \tfrac{A^2_h}{m^2_h},\, \max_{h \in \H} \tfrac{A^2_h}{M^2_h}\right]:\, 
	s(J_m^{\lambda},\, J^{\lambda}_M) = 0\right\},
\end{equation*}
and note that $I \neq \emptyset$ is possible. If {\em FPIA} encounters $\lambda_k \in I \neq \emptyset$ at some $k = 0, 1, 
\ldots$, then the crucial step of {\em FPIA}, i.e.:
\begin{equation}
    \label{eq:fpia_crucialstep}
    \lambda_{k+1} := \tfrac{1}{s^2\left(J_m^{\lambda_k},\,J^{\lambda_k}_M\right)},
\end{equation}
yields $\lambda_{k+1} = \tfrac{1}{0}$, which is undefined. Numerical Example \ref{ex:fpia_blocked} illustrates this scenario.
\begin{example}
    \label{ex:fpia_blocked}
    Consider allocation for a population given in Table \ref{tab:fpia_blocked}. Initial value $\lambda_0 =6861.36 \in I = [2304,\, 
    6922.24]$ and therefore {\em FPIA} gets blocked since $\lambda_1 = \tfrac{1}{0}$ is undefined.
    \begin{table}[H]
        \caption{Details of {\em FPIA} performance for a population with four strata. Here, $\lambda_0 = 
        \tfrac{1}{s^2(\emptyset,\, 
        \emptyset)},\, \lambda_1 = \tfrac{1}{s^2(J_m^{\lambda_0},\, J_M^{\lambda_0})},\, \lambda^* = 
        \tfrac{1}{s^2(J_m^{\lambda^*},\, J_M^{\lambda^*})}$. Squares in cells indicate assignment of stratum $h \in \H = \{1, 2, 
        3, 4\}$ to $J_m^{\lambda}$ (\tmin) or $J_M^{\lambda}$ (\tmax) for a given $\lambda$.}
		\footnotesize
		\centering
		\label{tab:fpia_blocked}
		\begin{tabular}[t]{| l | r | r r | r r || c | c || c | S |}
			\hline
			\multicolumn{6}{|c||}{population} & \multicolumn{2}{c||}{{\em FPIA}} & \multicolumn{2}{c|}{optimum allocation} \\
			\hline
			$h$ & $A_h$ & $m_h$ & $M_h$ & $\tfrac{A_h^2}{m_h^2}$ & $\tfrac{A_h^2}{M_h^2}$ &
			$J_m^{\lambda_0} / J_M^{\lambda_0}$ & $J_m^{\lambda_1} / J_M^{\lambda_1}$ &
			$J_m^{\lambda^*} / J_M^{\lambda^*}$ & $\x^*$ \\
			\hline
			1 & 4160 & 5 & 50 & 692224 & 6922.24 & \tmax & \multirow{4}{*}{-} & & 44.35 \\
			2 & 240 & 5 & 50 & 2304 & 23.04 & \tmin & & \tmin & 5 \\ 
			3 & 530 & 5 & 50 & 11236 & 112.36 & & & & 5.65 \\ 
			4 & 40 & 5 & 50 & 64 & 0.64 & \tmin & & \tmin & 5 \\
			\hline
			\multicolumn{6}{|l||}{total sample size $n = 60$} & $\lambda_0 = 6861.36$ & $\lambda_1 = \tfrac{1}{0}$ & $\lambda^* 
			= 8798.44$ & \multicolumn{1}{c}{} \\
			\cline{1-9}			
		\end{tabular}
    \end{table}

    Figure \ref{fig:fpia_blocked} shows the graphs of the functions $g$ and $\phi$ for the problem considered in this example.
    \begin{figure}[H]
		\caption{Functions $g$ and $\phi$ for Example \ref{ex:fpia_blocked} of the allocation problem for which the {\em FPIA} 
		gets blocked.}
		\label{fig:fpia_blocked}
	\begin{center}
		\includegraphics[height=70mm,width=140mm]{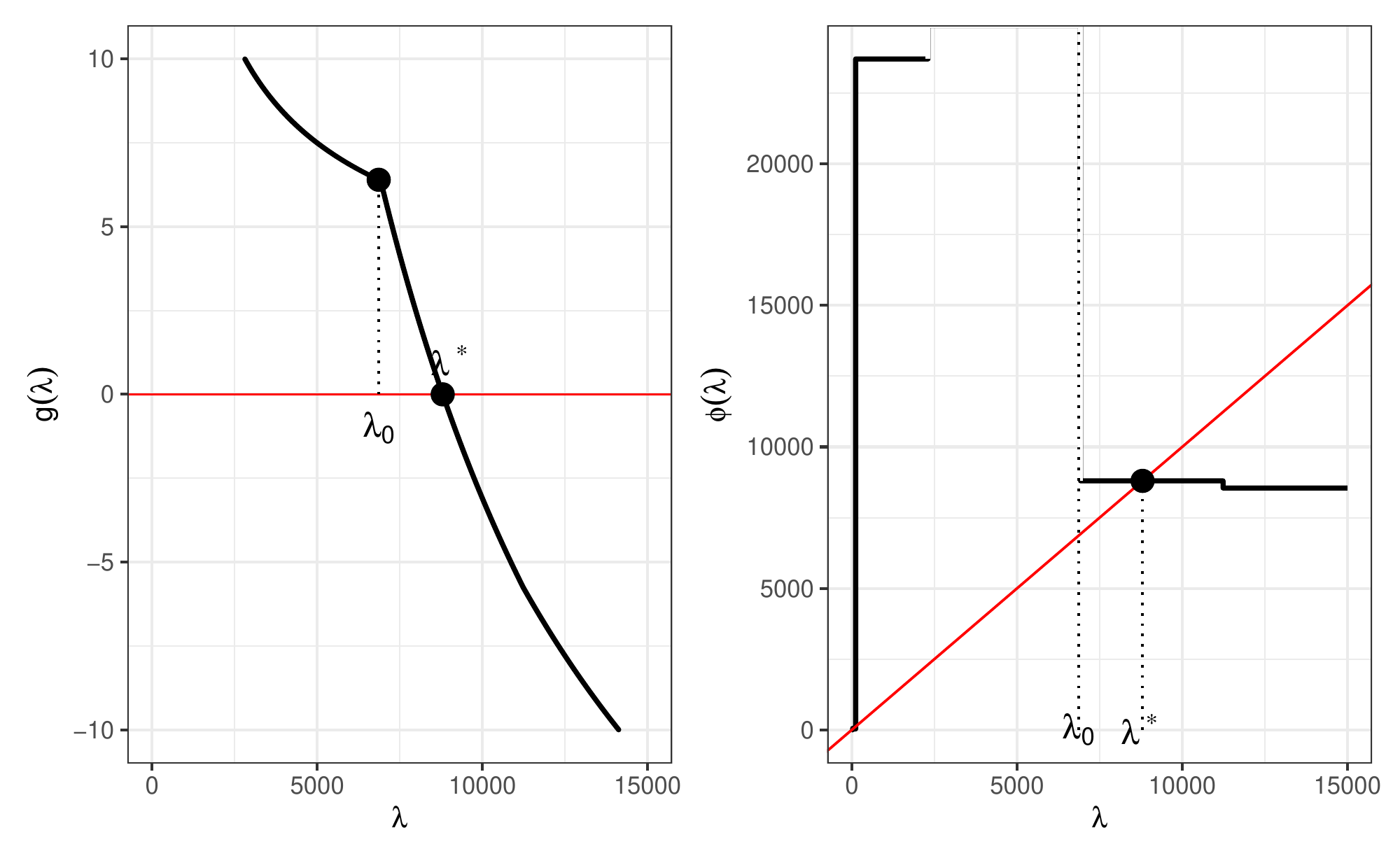}
	\end{center}
    \end{figure}
    These graphs show that $g$ is not differentiable, and $\phi$ has jump discontinuities. Moreover, there is an interval $I = 
    [2304,\,6922.24]$, such that $\phi$ is not well-defined on $I$.
\end{example}

We note that a simple remedy to avoid the blocking scenario in the case of a {\em regular} allocation is to change $\lambda$ 
to $\lambda^{new} = \tfrac{1}{\lambda}$, with a corresponding redefinition of all the objects in the original algorithm.

The next Example \ref{ex:fpia_oscillations} illustrates a situation when {\em FPIA} does not converge.

\begin{example}
    \label{ex:fpia_oscillations}
    Consider a population given in Table \ref{tab:fpia_oscillations}. Initial value of $\lambda_0 = 695.64$ causes lack of 
    convergence of the {\em FPIA} due to oscillations:
    \begin{equation*}
	\lambda_k =
	\begin{cases}
		1444,		& \ind{k = 1,\, 3,\, 5,\, \ldots} \\
		739.84,   & \ind{k = 2,\, 4,\, 6,\, \ldots.}
	\end{cases}
    \end{equation*}
	
    \begin{table}[H]
		\caption{Details of {\em FPIA} performance for a  population with four strata. Here, $\lambda_0 = 
		\tfrac{1}{s^2(\emptyset,\, \emptyset)}$, $\lambda_{k+1}$ is as in \eqref{eq:fpia_crucialstep}, and $\lambda^* = 
		\tfrac{1}{s^2(J_m^{\lambda^*},\, J_M^{\lambda^*})}$. Squares in cells indicate assignment of stratum $h \in \H = \{1, 2, 
		3, 4\}$ to $J_m^{\lambda}$ (\tmin) or $J_M^{\lambda}$ (\tmax) for a given $\lambda$.}
		\footnotesize
		\centering
		\label{tab:fpia_oscillations}
		\begin{tabular}[t]{| l | r | r r | r r || c | c | c | c || c | S |}
			\hline
			\multicolumn{6}{|c||}{population} & \multicolumn{4}{c||}{{\em FPIA}} & \multicolumn{2}{c|}{optimum allocation} \\
			\hline
			$h$ & $A_h$ & $m_h$ & $M_h$ & $\tfrac{A_h^2}{m_h^2}$ & $\tfrac{A_h^2}{M_h^2}$ & $J_m^{\lambda_0} / 
			J_M^{\lambda_0}$ & $J_m^{\lambda_1} / J_M^{\lambda_1}$ & $J_m^{\lambda_2} / J_M^{\lambda_2}$ & $\cdots$ & 
			$J_m^{\lambda^*} / J_M^{\lambda^*}$ & $\x^*$ \\
            \hline
			1 & 380 & 10 & 50 & 1444 & 57.76 & & \tmin & &  \multirow{4}{*}{$\cdots$} & & 13.1 \\	
			2 & 140 & 10 & 50 & 196 & 7.84 & \tmin & \tmin & \tmin & & \tmin & 10 \\ 
			3 & 230 & 10 & 50 & 529 & 21.16 & \tmin & \tmin & \tmin & & \tmin & 10 \\ 
			4 & 1360 & 10 & 50 & 18496 & 739.84 & \tmax & & \tmax & & & 46.9 \\
			\hline
			\multicolumn{6}{|l||}{total sample size $n = 80$} & $\lambda_0 = 695.64$ & $\lambda_1 = 1444$ & $\lambda_2 = 
			739.84$ & $\cdots$ & $\lambda^* = 841$ & \multicolumn{1}{c}{} \\
			\cline{1-11}			
		\end{tabular}
    \end{table}
 	
    Figure \ref{fig:fpia_oscillations} shows the graphs of the functions $g$ and $\phi$ for the problem considered in this 
    example.
    \begin{figure}[H]
		\caption{Functions $g$ and $\phi$ for Example \ref{ex:fpia_blocked} of the allocation problem for which the {\em FPIA} 
		does not converge.}
		\label{fig:fpia_oscillations}
		\begin{center}
			\includegraphics[height=70mm,width=140mm]{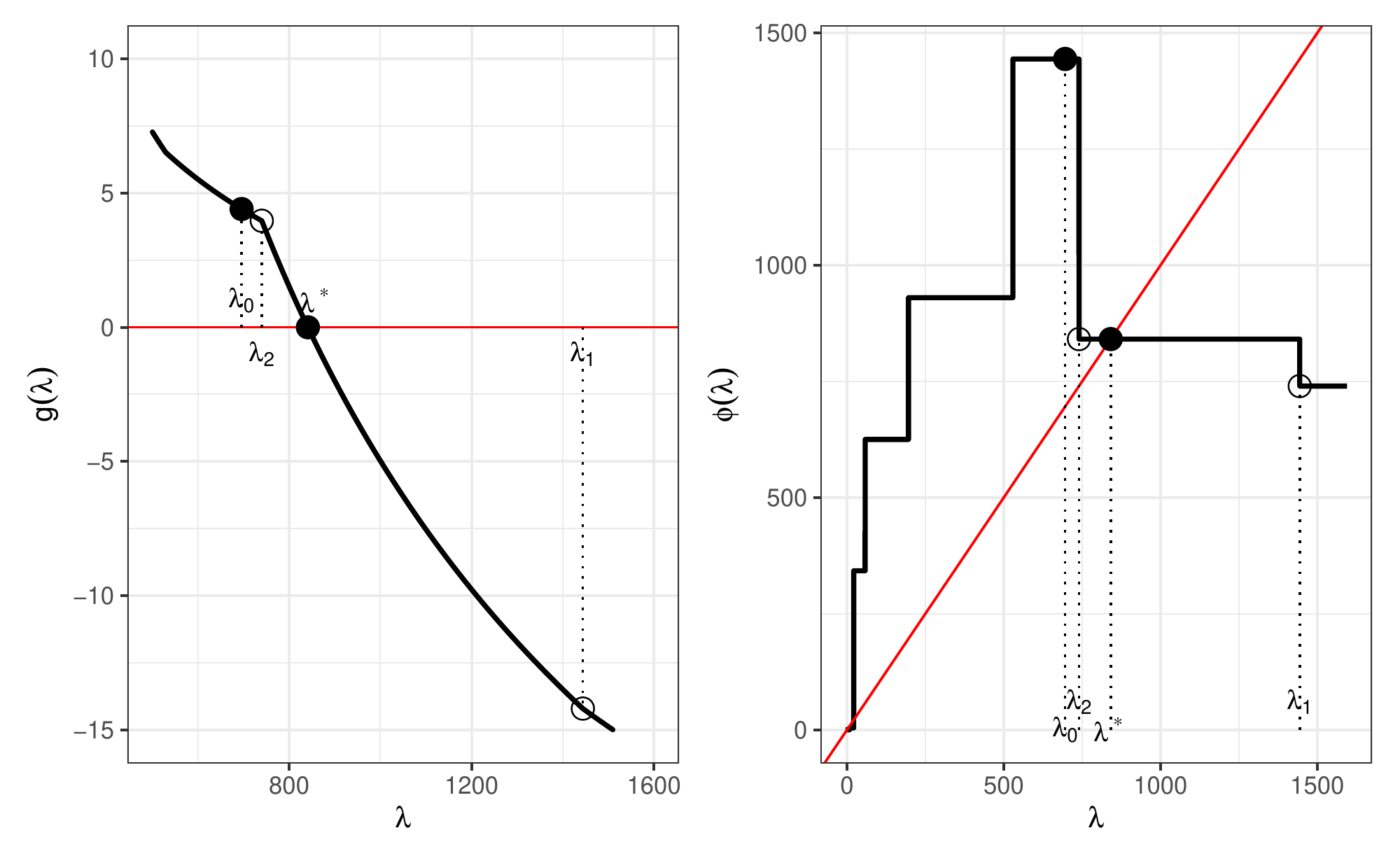}
		\end{center}
    \end{figure}
\end{example}

The issue of determining a proper starting point $\lambda_0$ was considered in MSW, with a recommendation to choose 
$\lambda_0 = \tfrac{1}{s^2(\emptyset,\, \emptyset)}$. Alternatively,  in case when $\tfrac{1}{s^2(\emptyset,\, \emptyset)}$ is 
not close enough to the optimal $\lambda^*$ (which is not known a priori), MSW suggests to first run several iterations of a 
root finding algorithm to get the starting point $\lambda_0$ for the {\em FPIA}.

\section{Appendix: Proof of Theorem \ref{th:optcond}} 
\label{app:optcon}

\begin{remark}
    \label{rem:prob_convex}
    Problem \ref{prob} is a convex optimization problem as its objective function $f: \R_+^{\card{\H}} \to \R_+$,
    \begin{equation}
		f(\x)  = \sum_{h \in \H} \tfrac{A_h^2}{x_h}, \label{eq:rem_prob_convex:f}
    \end{equation}
    and inequality constraint functions $g_h^m: \R_+^{\card{\H}} \to \R,\, g_h^M: \R_+^{\card{\H}} \to \R$, 
    \begin{align}
		&g_h^m(\x) = m_h - x_h,  \ind{h \in \H,} \label{eq:rem_prob_convex:gL} \\
		&g_h^M(\x) = x_h - M_h, \ind{h \in \H,} \label{eq:rem_prob_convex:gU}
    \end{align}	
    are convex functions, whilst the equality constraint function $w: \R_+^{\card{\H}} \to \R$, 
    \begin{equation*}
		w(\x) = \sum_{h \in \H} x_h - n \\
    \end{equation*}	
    is affine. More specifically, Problem \ref{prob} is a convex optimization problem of a particular type in which inequality 
    constraint functions \eqref{eq:rem_prob_convex:gL} - \eqref{eq:rem_prob_convex:gU} are affine. See Appendix 
    \ref{app:kkt} for the definition of the convex optimization problem.
\end{remark}

\begin{proof}[Proof of Theorem \ref{th:optcond}]
    We first prove that Problem \ref{prob} has a unique solution. The optimization Problem \ref{prob} is feasible since 
    requirements $m_h < M_h,\, h \in \H$, and $\sum_{h \in \H} m_h \leq n \leq \sum_{h \in \H} M_h$ ensure that the feasible 
    set $F := \{\x \in \R_+^{\card{\H}}: \text{ \eqref{eq:prob_cnst_eq} - \eqref{eq:prob_cnst_ineq} are all satisfied}\}$ is 
    non-empty. The objective function \eqref{eq:prob_obj} attains its minimum on $F$ since it is a continuous function and 
    $F$ is closed and bounded. Finally, uniqueness of the solution is due to strict convexity of the objective function on $F$.
	
    As explained in Remark \ref{rem:prob_convex}, Problem \ref{prob} is a convex optimization problem in which the inequality 
    constraint functions $g_h^m,\, g_h^M,\, h \in \H$ are affine. The optimal solution for such a problem can be identified 
    through the Karush-Kuhn-Tucker (KKT) conditions, in which case they are not only necessary but also sufficient; for further 
    references, see Appendix \ref{app:kkt}.
	
    The gradients of the objective function \eqref{eq:rem_prob_convex:f} and constraint functions 
    \eqref{eq:rem_prob_convex:gL} - \eqref{eq:rem_prob_convex:gU} are as follows:
    \begin{equation*}
		\nabla f(\x) = (-\tfrac{A_h^2}{x_h^2},\, h \in \H), \quad
		\nabla w(\x) = \underline 1, \quad 
		\nabla g_h^m(\x)  = -\nabla g_h^M(\x) = -\underline{1}_h, \quad h \in \H,
    \end{equation*}
    where, $\underline 1$ is a vector with all entries $1$ and $\underline{1}_h$ is a vector with all entries $0$ except the entry with the label $h$, which is $1$. Hence, the KKT conditions \eqref{KKT} for Problem \ref{prob} assume the form 
    \begin{gather}
		-\tfrac{A_h^2}{{(x^*_h)}^2} + \lambda - \mu_h^m + \mu_h^M = 0, \ind{h \in \H,} \label{eq:kkt_prob_stat} \\
		\sum_{h \in \H} x_h^*  - n = 0 \label{eq:kkt_prob_n}, \\
		m_h \leq x^*_h \leq M_h, \ind{h \in \H,} \label{eq:kkt_prob_ineq} \\
		\mu^m_h(m_h - x^*_h) = 0, \ind{h \in \H,} \label{eq:kkt_prob_compl1} \\
		\mu_h^M(x^*_h - M_h)  = 0, \ind{h \in \H.} \label{eq:kkt_prob_compl2}
    \end{gather}
    To prove Theorem \ref{th:optcond}, it suffices to show that for $\x^* = \x^{(\L^*,\, \U^*)}$ with $\L^*,\, \U^*$ satisfying 
    conditions of CASE I or CASE II, there exist $\lambda \in \R$ and $\mu^m_h,\, \mu^M_h \ge 0,\, h \in \H$, such that 
	\eqref{eq:kkt_prob_stat} - \eqref{eq:kkt_prob_compl2} hold. It should also be noted that the requirement $m_h < M_h,\, h 
	\in \H$, guarantees that $\L^*$ and $\U^*$ defined in \eqref{eq:optcond} and \eqref{eq:optcond_vertex} are disjoint. 
	Therefore, $\x^{(\L^*,\, \U^*)}$ is well-defined according to Definition \ref{def:LU}.
	
    \begin{enumerate}[wide, labelindent=0pt, leftmargin=*]
		\item[CASE I:] 
		Take $\x^*=\x^{(\L^*,\, \U^*)}$ with $\L^*$ and $\U^*$ as in \eqref{eq:optcond}. Then, \eqref{eq:kkt_prob_n} is clearly 
		met after referring to \eqref{eq:LU} and \eqref{eq:s}, while \eqref{eq:kkt_prob_ineq} follows directly from \eqref{eq:LU} 
		and \eqref{eq:optcond}, since \eqref{eq:optcond} for $h \in \H \setminus (\L^* \cup \U^*)$ specifically implies $m_h < 
		A_h\, s(\L^*,\, \U^*) < M_h$. Take $\lambda = \tfrac{1}{s^2(\L^*,\, \U^*)}$ and
		\begin{equation}
			\label{eq:optcond_proof:mu}
			\mu_h^m =
			\begin{cases}
				\lambda - \tfrac{A_h^2}{m_h^2}, & h \in \L^* \\
				0, & h \in \H \setminus \L^*,
			\end{cases}
			\qquad \qquad 
			\mu_h^M =
			\begin{cases}
				\tfrac{A_h^2}{M_h^2} - \lambda, & h \in \U^* \\
				0, & h \in \H \setminus \U^*.
			\end{cases}
		\end{equation}
		Note that \eqref{eq:optcond} along with requirement $n \geq \sum_{h \in \H} m_h$ (the latter needed if $\U^* = 
		\emptyset$) ensure $s(\L^*,\, \U^*) > 0$, whilst \eqref{eq:optcond} alone implies $\mu^m_h,\, \mu^M_h \ge 0,\, h \in 
		\H$. After referring to \eqref{eq:LU}, it is a matter of simple algebra to verify \eqref{eq:kkt_prob_stat}, 
		\eqref{eq:kkt_prob_compl1} and \eqref{eq:kkt_prob_compl2} for $\lambda,\, \mu^m_h,\, \mu^M_h,\, h \in \H$ defined 
		above.
	\item[CASE II:]
	Take $\x^* = \x^{(\L^*,\, \U^*)}$ with $\L^*,\, \U^*$ satisfying \eqref{eq:optcond_vertex} and \eqref{eq:optcond_vertex_n}. 
	Then, condition \eqref{eq:kkt_prob_n} becomes \eqref{eq:optcond_vertex_n}, while \eqref{eq:kkt_prob_ineq} is trivially 
	met due to \eqref{eq:LU}. Assume that $\L^* \neq \emptyset$ and $\U^* \neq \emptyset$ (for empty $\L^*$ or $\U^*$, 
	\eqref{eq:kkt_prob_stat}, \eqref {eq:kkt_prob_compl1} and \eqref{eq:kkt_prob_compl2} are trivially met). Take an arbitrary 
	$\tilde s > 0$ such that
	\begin{equation}
		\label{eq:optcond_vertex_proof}
		\tilde s \in \left[\max_{h \in \U^*} \tfrac{M_h}{A_h},\; \min_{h \in \L^*} \tfrac{m_h}{A_h}\right].
	\end{equation}
	Note that \eqref{eq:optcond_vertex} ensures that the interval above is well-defined. Let $\lambda = \tfrac{1}{\tilde s^2}$ 
	and $\mu_h^m,\, \mu_h^M,\, h \in \H$ be as in \eqref{eq:optcond_proof:mu}. Note that \eqref{eq:optcond_vertex_proof} 
	ensures that $\mu_h^m,\, \mu_h^M \geq 0$ for all $h \in \H$. Then it is easy to check, similarly as in CASE I, that 
	\eqref{eq:kkt_prob_stat}, \eqref{eq:kkt_prob_compl1} and \eqref{eq:kkt_prob_compl2} are satisfied.
    \end{enumerate}
\end{proof}

\section{Appendix: Auxiliary lemmas and proof of Theorem \ref{th:rnabox}}
\label{app:rnabox_proof}

\subsection{Notation}
\label{app:notation}

Throughout the Appendix \ref{app:rnabox_proof}, by $\U_r,\, \L_r,\, \Lt_r$, we denote sets $\U,\, \L,\, \Lt$ respectively, as 
they are in the $r$-th iteration of {\em RNABOX} algorithm after  {\footnotesize Step \ref{alg:rnabox:L}} and before 
{\footnotesize Step \ref{alg:rnabox:if}}. The iteration index $r$ takes on values from set $\{1, \ldots, r^*\}$, where $r^* \geq 1$ 
indicates the final iteration of the algorithm. Under this notation, we have $\L_1 = \emptyset$ and in general, for subsequent 
iterations, if any (i.e. if $r^* \geq 2$), we get
\begin{equation}
    \label{eq:rnabox_Lsum}
    \L_r = \L_{r-1} \cup \Lt_{r-1} = \bigcup\limits_{i = 1}^{r-1} \Lt_i, \ind{r = 2, \ldots, r^*.}
\end{equation}

\bigskip
As {\em RNABOX} iterates, objects denoted by symbols $n$ and $\H$ are being modified. However, in this Appendix 
\ref{app:rnabox_proof}, whenever we refer to $n$ and $\H$, they always denote the unmodified total sample size and the set 
of strata labels as in the input of {\em RNABOX}. In particular, this is also related to set function $s$ (defined in 
\eqref{eq:s}) which depends on $n$ and $\H$.

\bigskip
For convenient notation, for any $\A \subseteq \H$ and any set of real numbers $z_h,\, h \in \A$, we denote
\begin{equation*}
    z_\A = \sum_{h \in \A} z_h.
\end{equation*}

\subsection{Auxiliary remarks and lemmas}

We start with a lemma describing important  monotonicity properties of function $s$.
\begin{lemma}
    \label{lem:smono} 
    Let $\A \subseteq \B \subseteq \H$ and  $\C \subseteq \D \subseteq \H$.
    \begin{enumerate}
		\item If $\B \cup \D \subsetneq \H$ and $\B \cap \D = \emptyset$, then
		\begin{equation}
			\label{eq:lem:smono:1}
            s(\A,\, \C) \geq s(\B,\, \D) \quad \Leftrightarrow \quad s(\A,\, \C)(A_{\B \setminus \A} + A_{\D \setminus \C}) \leq m_{\B 
            \setminus \A} + M_{\D \setminus \C.}
		\end{equation}
		\item If $\A \cup \D \subsetneq \H$, $\A \cap \D = \emptyset$, $\B \cup \C \subsetneq \H$, $\B \cap \C = \emptyset$, then
        \begin{equation}
        	\label{eq:lem:smono:2}
            s(\A,\, \D) \geq s(\B,\, \C) \quad \Leftrightarrow \quad s(\A,\, \D)(A_{\B \setminus \A} - A_{\D \setminus \C}) \leq m_{\B 
            \setminus \A} - M_{\D \setminus \C.}
		\end{equation}
    \end{enumerate}
\end{lemma}

\begin{proof}
    Clearly, for any $\alpha \in \R,\, \beta \in \R, \delta \in \R$, $\gamma > 0,\, \gamma + \delta > 0$, we have 
    \begin{equation}
		\label{eq:lem:proof:smono:gene}
		\tfrac{\alpha+\beta}{\gamma+\delta} \geq \tfrac{\alpha}{\gamma} \quad \Leftrightarrow \quad	 
		\tfrac{\alpha+\beta}{\gamma+\delta} \delta \leq \beta.
    \end{equation}
    To prove \eqref{eq:lem:smono:1}, take
    \begin{align*}
		\alpha	  &= n - m_{\B} - M_{\D}       & \beta  &= m_{\B \setminus \A} + M_{\D \setminus \C} \\
		\gamma &= A_{\H} - A_{\B \cup \D} & \delta &= A_{\B \setminus \A} + A_{\D \setminus \C}.
    \end{align*}
    Then, $\tfrac{\alpha}{\gamma} = s(\B,\, \D),\, \tfrac{\alpha + \beta}{\gamma + \delta} = s(\A,\, \C)$, and hence 
    \eqref{eq:lem:smono:1} holds as an immediate consequence of \eqref{eq:lem:proof:smono:gene}. \\
    Similarly for  \eqref{eq:lem:smono:2}, take
    \begin{align*}
		\alpha	  &= n - m_{\B} - M_{\C} 	   & \beta  &= m_{\B \setminus \A} - M_{\D \setminus \C}  \\
		\gamma &= A_{\H} - A_{\B \cup \C} & \delta &= A_{\B \setminus \A} - A_{\D \setminus \C},
    \end{align*}
    and note that $\gamma + \delta = A_{\H} - A_{\B \cup \C} + A_{\B \setminus \A} - A_{\D \setminus \C} = A_{\H} - A_{\B} - 
    A_{\C} + A_{\B} - A_{\A} - A_{\D} + A_{\C} = A_{\H} - A_{\A \cup \D} > 0$ due to the assumptions made for $\A,\, \D,\, \B,\, 
    \C$, and $A_h > 0,\, h \in \H$. Then, $\tfrac{\alpha}{\gamma} = s(\B,\, \C)$, $\tfrac{\alpha + \beta}{\gamma + \delta} = s(\A,\, 
    \D)$, and hence \eqref{eq:lem:smono:2} holds as an immediate consequence of \eqref{eq:lem:proof:smono:gene}.
\end{proof}

The remark below describes some relations between sets $\L_r$ and $\U_r,\, r = 1, \ldots, r^* \geq 1$, appearing in {\em 
RNABOX} algorithm. These relations are particularly important for understanding computations involving the  set function $s$ 
(recall, that it is defined only for such two disjoint sets, the union of which is a proper subset of $\H$).

\begin{remark}
    \label{rem:rnabox_LcupU} 
    For $r^* \geq 1$,
    \begin{equation}
		\label{eq:rem:rnabox_LU_disjoint}
    	\L_r \cap \U_r = \emptyset, \ind{r = 1, \ldots, r^*,}
    \end{equation}
    and for $r^* \geq 2$,
    \begin{equation}
		\label{eq:rem:rnabox_LcupU:any}
		\L_r \cup \U_r \subsetneq \H, \ind{r = 1, \ldots, r^* - 1.}
    \end{equation}
    Moreover, let $\x^*$ be as in {\footnotesize Step \ref{alg:rnabox:ret}} of {\em RNABOX} algorithm. Then, for $r^*\ge 1$, 
    \begin{equation}
		\label{eq:rem:rnabox_LcupU:reg}
		\L_{r^*} \cup \U_{r^*} \subsetneq \H \quad \Leftrightarrow \quad \mathbf x^* \mbox{ is a {\em regular} allocation,}
    \end{equation}
    and
    \begin{equation}
		\label{eq:rem:rnabox_LcupU:ver}
		\L_{r^*} \cup \U_{r^*} = \H \quad \Leftrightarrow \quad \mathbf x^*\mbox{is a {\em vertex} allocation.}
    \end{equation} 
\end{remark}

\begin{proof}
    From the definition of set $\U$ in {\footnotesize Step \ref{alg:rnabox:rna}} of {\em RNABOX}, for $r^* \geq 1$,
    \begin{equation}
		\label{eq:rem:rnabox_LcupU:proof:Usubset}
    	\U_r \subseteq \H \setminus \L_r, \ind{r = 1, \ldots, r^*,}
    \end{equation}
    which proves \eqref{eq:rem:rnabox_LU_disjoint}. Following \eqref{eq:rnabox_Lsum}, for $r^* \geq 2$,
     \begin{equation}
		\label{eq:rem:rnabox_LcupU:proof:Lsubset}
    	\L_r = \bigcup\limits_{i = 1}^{r - 1} \Lt_i \subseteq \H, \ind{r = 2, \ldots, r^*,}
    \end{equation}
    where the inclusion is due to definition of set $\Lt$ in {\footnotesize Step \ref{alg:rnabox:L}} of {\em RNABOX}, i.e. $\Lt_r 
    \subseteq \H \setminus (\L_r \cup \U_r)$ for $r = 1, \ldots, r^*$. Inclusions 
    \eqref{eq:rem:rnabox_LcupU:proof:Usubset}, \eqref{eq:rem:rnabox_LcupU:proof:Lsubset} with $\L_1 = \emptyset$ imply
    \begin{equation}
		\label{eq:rem:rnabox_LcupU:proof:LU}
		\L_r \cup \U_r \subseteq \H, \ind{r = 1, \ldots, r^* \geq 1.}
    \end{equation}
    Given that $r^* \geq 2$, {\footnotesize Step \ref{alg:rnabox:if}} of the algorithm ensures that set $\Lt_r \subseteq \H   
    \setminus (\L_r \cup \U_r)$ is non-empty for $r = 1, \ldots, r^* -1$, which implies $\L_r \cup \U_r \neq \H$. This fact 
    combined with \eqref{eq:rem:rnabox_LcupU:proof:LU} gives \eqref{eq:rem:rnabox_LcupU:any}.
    Equivalences \eqref{eq:rem:rnabox_LcupU:reg} and \eqref{eq:rem:rnabox_LcupU:ver} hold trivially after referring to 
    Definition \ref{def:vertex_regular} of {\em regular} and {\em vertex} allocations.
\end{proof}

The following two remarks summarize some important facts arising from {\footnotesize Step \ref{alg:rnabox:rna}} of {\em 
RNABOX} algorithm. These facts will serve as starting points for most of the proofs presented in this section.

\begin{remark}
    \label{rem:rnabox_rna_x} 
    In each iteration $r = 1, \ldots, r^* \geq 1$, of RNABOX algorithm, a vector $(x^{**}_h,\, h \in \H \setminus \L_r)$ obtained in 
    {\footnotesize Step \ref{alg:rnabox:rna}}, has the elements of the form
    \begin{equation}
		\label{eq:rem:rnabox_rna_x}
		x^{**}_h =
		\begin{cases}
			M_h,						& \qquad{h \in \U_r \subseteq \H \setminus \L_r} \\
			A_h\, s(\L_r,\, \U_r) < M_h,	& \ind{h \in \H \setminus (\L_r \cup \U_r),}
		\end{cases}
    \end{equation}
    where the set function $s$ is defined in \eqref{eq:s}. Equation \eqref{eq:rem:rnabox_rna_x} is a direct consequence of 
    Theorem \ref{th:optcond_one_sided}.
\end{remark}

\begin{remark}
    \label{rem:rnabox_UUL}
    Remark \ref{rem:rnabox_rna_x} together with Theorem \ref{th:optcond_one_sided}, for $r^* \geq 2$ yield
    \begin{equation}
		\label{eq:rem:rnabox_UUL:U}
		\U_r = \{h \in \H \setminus \L_r:\, A_h\, s(\L_r,\, \U_r) \geq M_h\}, \ind{r = 1, \ldots, r^* - 1,}
    \end{equation}
    whilst for $r^*\ge 1$,
    \begin{equation}
		\label{eq:rem:rnabox_UUL:U_reg}
		\U_{r^*} = \{h \in \H \setminus \L_{r^*}:\, A_h\, s(\L_{r^*},\, \U_{r^*}) \geq M_h\},
    \end{equation}
    if and only if $\x^*$ (computed at {\footnotesize Step \ref{alg:rnabox:ret}} of {\em RNABOX} algorithm) is: a {\em regular} 
    allocation or a {\em vertex} allocation with $\L_{r^*} = \H$.
	
    Moreover, for $r^* \geq 1$,
    \begin{equation}
		\label{eq:rem:rnabox_UUL:L}
		\Lt_r = \{h \in \H \setminus (\L_r \cup \U_r):\, A_h\, s(\L_r,\, \U_r) \leq m_h\}, \ind{r = 1, \ldots, r^*.}
    \end{equation}
\end{remark}
Note that in Remark \ref{rem:rnabox_UUL}, function $s$ is well-defined due to Remark \ref{rem:rnabox_LcupU}. The need to 
limit the scope of \eqref{eq:rem:rnabox_UUL:U_reg} to {\em regular} allocations only, is dictated by the fact that in the case 
of a {\em vertex} allocation we have $\L_{r^*} \cup \U_{r^*} = \H$ (see \eqref{eq:rem:rnabox_LcupU:ver}) and therefore 
$s(\L_{r^*},\, \U_{r^*})$ is not well-defined.

\bigskip
Lemma \ref{lem:U} and Lemma \ref{lem:s_in_rnabox} reveal certain monotonicity properties of sequence $(\U_r)_{r = 
1}^{r^*}$ and sequence $\left(s(\L_r,\, \U_r)\right)_{r = 1}^{r^*}$, respectively. These properties will play a crucial role in 
proving Theorem \ref{th:rnabox}.

\begin{lemma}
    \label{lem:U}
    Sequence $(\U_r)_{r = 1}^{r^*}$ is non-increasing, that is, for $r^* \geq 2$,
    \begin{equation}
		\label{eq:lem:U}
		\U_r \supseteq \U_{r+1}, \ind{r = 1, \ldots, r^* - 1.}
    \end{equation}
\end{lemma}

\begin{proof}
    Let $r^* \geq 2$ and $r = 1, \ldots, r^* - 1$. Then, by \eqref{eq:rem:rnabox_LcupU:any}, $\L_r \cup \U_r  
    \subsetneq \H$. Following \eqref{eq:rem:rnabox_UUL:U}, the domain of discourse for $\U_r$ is  $\H \setminus \L_r$, and in 
    fact it is $\H \setminus (\L_r \cup \Lt_r) = \H \setminus \L_{r+1}$, since $\U_r \not\subset \Lt_r$ as ensured by 
    {\footnotesize Step \ref{alg:rnabox:L}} of {\em RNABOX}. That is, both $\U_r$ and $\U_{r+1}$ have essentially the same 
    domain of discourse, which is $\H \setminus \L_{r+1}$. Given this fact and the form of the set-builder predicate in 
    \eqref{eq:rem:rnabox_UUL:U} - \eqref{eq:rem:rnabox_UUL:U_reg} as well as equality $\U_{r^*} = \H \setminus \L_{r^*}$ for 
    the case when $\x^*$ is a  {\em vertex} allocation (for which \eqref{eq:rem:rnabox_UUL:U_reg} does not apply), we 
    conclude that only one of the following two distinct cases is possible: $\U_r \supseteq \U_{r+1}$ or $\U_r \subsetneq 
    \U_{r+1}$. 
	
    The proof is by contradiction, that is, assume that \eqref{eq:lem:U} does not hold. Therefore, in view of the above 
    observation, there exists $r \in \{1, \ldots, r^* - 1\}$ such that $\U_r \subsetneq \U_{r+1}$. Then,
    \begin{equation}
		\label{eq:lem:U:proof:assumption}
		\emptyset \neq (\U_{r+1} \setminus \U_r) \subsetneq \H \setminus (\L_r \cup \U_r),
    \end{equation}
    and hence, due to \eqref{eq:rem:rnabox_rna_x},
    \begin{equation}
		\label{eq:lem:U:proof:Uc}
		A_h\, s(\L_r,\, \U_r) < M_h, \ind{h \in \U_{r+1} \setminus \U_r.} 
    \end{equation}
    On the other hand, from \eqref{eq:rem:rnabox_UUL:L},
    \begin{equation}
		\label{eq:lem:U:proof:L}
		A_h\, s(\L_r,\, \U_r) \leq m_h, \ind{h \in \Lt_r.} 
    \end{equation}
    Summing sidewise: \eqref{eq:lem:U:proof:Uc} over $h \in \U_{r+1} \setminus \U_r$, \eqref{eq:lem:U:proof:L} over $h \in \Lt_r$, and then all together, we get
    \begin{equation}
		\label{eq:lem:U:proof:smono2}
		s(\L_r,\, \U_r) (A_{\Lt_r} + A_{\U_{r+1} \setminus \U_r}) < m_{\Lt_r} + M_{\U_{r+1} \setminus \U_r}.
    \end{equation}
    \begin{enumerate}[wide, labelindent=0pt, leftmargin=*, font=\itshape]
    	
    	\item[Vector $\x^*$ is a {\em regular} allocation:]
    	In this case, following Remark \ref{rem:rnabox_LcupU}, we see that inequality \eqref{eq:lem:U:proof:smono2} is the 
    	right-hand side of equivalence \eqref{eq:lem:smono:1} with
    	\begin{align}
    		\A &= \L_r \subseteq (\L_r \cup \Lt_r) = \L_{r+1} = \B \subsetneq \H, \\
    		\C &= \U_r \subsetneq \U_{r+1} = \D \subsetneq \H.
    	\end{align}
    	Then, following Lemma \ref{lem:smono}, inequality \eqref{eq:lem:U:proof:smono2} is equivalent to
        \begin{equation}
    		\label{eq:lem:U:proof:smono1}
    		s(\L_r, \U_r) > s(\L_{r+1}, \U_{r+1}).
    	\end{equation}
    	Combining
	    \begin{equation}
			\label{eq:lem:U:proof:step_rna_Dr1}
			s(\L_{r+1},\, \U_{r+1}) \geq \tfrac{M_h}{A_h},  \ind{h \in \U_{r+1},}
    	\end{equation}
    	(it follows from \eqref{eq:rem:rnabox_UUL:U}-\eqref{eq:rem:rnabox_UUL:U_reg}) with inequalities 
    	\eqref{eq:lem:U:proof:smono1} and \eqref{eq:lem:U:proof:Uc}, we get the contradiction
	    \begin{equation}
			\tfrac{M_h}{A_h} > s(\L_r,\, \U_r) > s(\L_{r+1},\, \U_{r+1}) \geq \tfrac{M_h}{A_h}, \ind{h \in \U_{r+1} \setminus \U_r.}
		\end{equation}
    	Therefore, \eqref{eq:lem:U} holds true, given that $\L_{r^*} \cup \U_{r^*} \subsetneq \H$.
    
    	\item[Vector $\x^*$ is a {\em vertex} allocation:]
    	Since $\L_{r+1} \cup \U_{r+1} \subsetneq \H$ for $r = 1, \ldots, r^* -2$, the proof of \eqref{eq:lem:U} for such $r$ is 
    	identical to the 
    	proof for the case of {\em regular} allocation. Hence, we only need to show that \eqref{eq:lem:U} holds for $r = r^* -1$. 
    	For this purpose, we will exploit inequality \eqref{eq:lem:U:proof:smono2}, which in view of Definition \ref{def:s} of set 
    	function $s$, assumes the following form for $r = r^* - 1$,
    	\begin{equation}
			\label{eq:lem:U:proof:smono_vertex}
			\tfrac{n - m_{\L_{r^*-1}} - M_{\U_{r^*-1}}}{A_{\H} - A_{\L_{r^*-1} \cup \U_{r^*-1}}}\, (A_{\Lt_{r^*-1}} + A_{\U_{r^*} 	
			\setminus \U_{r^*-1}}) <  m_{\Lt_{r^*-1}} + M_{\U_{r^*} \setminus \U_{r^*-1}}.
    	\end{equation}
    	Since $A_{\H} - A_{\L_{r^*-1} \cup \U_{r^*-1}} = A_{\Lt_{r^*-1}} + A_{\U_{r^*} \setminus \U_{r^*-1}}$, for $\L_{r^*} \cup 
    	\U_{r^*} = \H$, inequality \eqref{eq:lem:U:proof:smono_vertex} simplifies to
    	\begin{equation}
			n <  m_{\Lt_{r^*-1}} + m_{\L_{r^*-1}} + M_{\U_{r^*} \setminus \U_{r^*-1}} + M_{\U_{r^*-1}} = m_{\L_{r^*}} + 
			M_{\U_{r^*}} = n,
	    \end{equation}
    	which is a contradiction. Note that the last equality follows from {\footnotesize Step \ref{alg:rnabox:rna}} of the {\em 
    	RNABOX} after referring to \eqref{eq:th:optcond_one_sided:M} and using the fact that $\U_{r^*} = \H \setminus 
    	\L_{r^*}$ for a {\em vertex} allocation. Therefore, \eqref{eq:lem:U} holds true also for $\L_{r^*} \cup \U_{r^*} = \H$.
    
        \end{enumerate}
\end{proof}

\begin{lemma}
    \label{lem:s_in_rnabox}
    Let $r^*\ge 3$. Then
    \begin{equation}
		\label{eq:lem:s_in_rnabox_any}
		s(\L_r,\, \U_r) \geq s(\L_{r+1},\, \U_{r+1}), \ind{r = 1, \ldots, r^* -2.}
    \end{equation}
    Moreover, if $\x^*$ (computed at {\footnotesize Step \ref{alg:rnabox:ret}} of {\em RNABOX} 
    algorithm) is a {\em regular} allocation and $r^* \geq 2$, then
    \begin{equation}
		\label{eq:lem:s_in_rnabox_reg}
		s(\L_{r^* -1},\, \U_{r^* -1}) \geq s(\L_{r^*},\, \U_{r^*}).
    \end{equation}
\end{lemma}

\begin{proof}
    We first prove \eqref{eq:lem:s_in_rnabox_any}. Let $r^* \geq 3$ and $r = 1, \ldots, r^* - 2$. Following Lemma \ref{lem:U} 
    and using \eqref{eq:rem:rnabox_UUL:U},
    \begin{equation}
		\label{eq:lem:s_in_rnabox:proof:U}
		A_h\, s(\L_r,\, \U_r) \geq M_h, \ind{h \in \U_r \setminus \U_{r+1}.}
    \end{equation}
    On the other hand, from \eqref{eq:rem:rnabox_UUL:L},
    \begin{equation}
		\label{eq:lem:s_in_rnabox:proof:L}
		A_h\, s(\L_r,\, \U_r) \leq m_h, \ind{h \in \Lt_r.} 
    \end{equation}
    Multiplying both sides of inequality \eqref{eq:lem:s_in_rnabox:proof:U} by $-1$, summing it sidewise over $h \in \U_r 
    \setminus \U_{r+1}$ and then adding it to \eqref{eq:lem:s_in_rnabox:proof:L}, which is previously summed sidewise over $h 
    \in \Lt_r$, we get
    \begin{equation}
		\label{eq:lem:s_in_rnabox:proof:smono1}
		s(\L_r,\, \U_r) (A_{\Lt_r} - A_{\U_r \setminus \U_{r+1}}) \leq m_{\Lt_r} - M_{\U_r \setminus \U_{r+1}}.
    \end{equation}
    Relation \eqref{eq:lem:s_in_rnabox:proof:smono1} is the second inequality in \eqref{eq:lem:smono:2} with
    \begin{align}
		\A &= \L_r \subsetneq (\L_r \cup \Lt_r) = \L_{r+1} = \B \subsetneq \H, \\
		\C &= \U_{r+1} \subseteq \U_r = \D \subsetneq \H.
    \end{align}
    Based on Remark \ref{rem:rnabox_LcupU}, we see that $\A \cup \D \subsetneq \H$, $\A \cap \D = \emptyset$, and $\B \cup 
    \C \subsetneq \H$, $\B \cap \C = \emptyset$, and thus the first inequality in \eqref{eq:lem:smono:2} follows, that is
    \begin{equation}
		\label{eq:lem:s_in_rnabox:proof:smono}
		s(\L_r,\, \U_r) \geq s(\L_{r+1},\, \U_{r+1}).
    \end{equation}	
    Hence \eqref{eq:lem:s_in_rnabox_any} is proved. 
	
	If $\x^*$ is a {\em regular} allocation, in view of Remark \ref{rem:rnabox_LcupU}, the same reasoning leading to inequality 
	\eqref{eq:lem:s_in_rnabox:proof:smono} clearly remains valid for $r = r^* -1,\, r^* \geq 2$.
\end{proof}

\subsection{Proof of  Theorem \ref{th:rnabox}}

To prove Theorem \ref{th:rnabox}, we have to show that:
\begin{enumerate}[label=(\Roman*), nosep, labelsep=5pt]
	\item the algorithm terminates in a finite number of iterations, i.e. $r^* < \infty$, \label{rnabox_proof_I}
	\item the solution computed at $r^*$ is optimal. \label{rnabox_proof_II}
\end{enumerate}
	
The proof of part \ref{rnabox_proof_I} is relatively straightforward. In every iteration $r = 1, \ldots, r^* - 1,\, r^* \geq 2$, the 
set of strata labels $\H$ is reduced by subtracting $\Lt_r$. Therefore, $r^* \leq \card{\H} + 1 < \infty$, where $r^* = 
\card{\H} + 1$ if and only if $\card{\Lt_r} = 1,\, r = 1, \ldots, r^* - 1$. In words, the algorithm stops after at most $\card{\H} + 1$ 
iterations.

In order to prove part \ref{rnabox_proof_II}, following Theorem \ref{th:optcond} and Remark \ref{rem:optcond:reg_ver}, it 
suffices to show that when $\x^*$ (computed at {\footnotesize Step \ref{alg:rnabox:ret}} of {\em RNABOX} algorithm) is a 
{\em regular} allocation, for all $h \in \H$,
\begin{gather}
	h \in \L_{r^*} \quad \Leftrightarrow \quad s(\L_{r^*},\, \U_{r^*}) \leq \tfrac{m_h}{A_h},  
	\label{eq:th:rnabox:proof:optcond_L} \\
	h \in \U_{r^*} \quad \Leftrightarrow \quad s(\L_{r^*},\, \U_{r^*}) \geq \tfrac{M_h}{A_h}, 
	\label{eq:th:rnabox:proof:optcond_U}
\end{gather}
and when $\x^*$ is a {\em vertex} allocation
\begin{gather}
	\max_{h \in \U_{r^*}} \tfrac{M_h}{A_h} \leq \min_{h \in \L_{r^*}} \tfrac{m_h}{A_h}, \quad \text{when } \U_{r^*} \neq 
	\emptyset \text{ and } \L_{r^*} \neq \emptyset, \label{eq:th:rnabox:proof:optcond_LU_ver} \\
	m_{\L_{r^*}} + M_{\U_{r^*}} = n. \label{eq:th:rnabox:proof:optcond_LUn_ver}
\end{gather}
   
\begin{enumerate}[wide, labelindent=0pt, leftmargin=*, font=\itshape]
   	\item[Vector $\x^*$ is a {\em regular} allocation:]
   	Note that Remark  \ref{rem:rnabox_LcupU}, implies that $s(\L_r,\, \U_r)$ is well-defined. We start with equivalence 
   	\eqref{eq:th:rnabox:proof:optcond_L}. \\
    {\em Necessity:} For $r^* = 1$, we have $\L_{r^*} = \emptyset$ and hence, the right-hand side of equivalence     
    \eqref{eq:th:rnabox:proof:optcond_L} is trivially met. Let $r^* \geq 2$, and $h \in \L_{r^*} = \bigcup\limits_{r = 1}^{r^*-1} 
    \Lt_r$. Thus, $h \in \Lt_r$ for some $r \in \{1,\ldots,r^*-1\}$ and then, due to \eqref{eq:rem:rnabox_UUL:L}, we have 
    $s(\L_r,\, \U_r) \leq \tfrac{m_h}{A_h}$. Consequently, \eqref{eq:lem:s_in_rnabox_any} with
    \eqref{eq:lem:s_in_rnabox_reg} yield $s(\L_{r^*},\, \U_{r^*}) \leq \tfrac{m_h}{A_h}$. \\
    {\em Sufficiency:} Since $\Lt_{r^*} = \emptyset$, \eqref{eq:rem:rnabox_UUL:L} implies
    \begin{equation}
		\label{eq:th:rnabox:proof:L}
		\{h \in \H \setminus (\L_{r^*} \cup \U_{r^*}):\, s(\L_{r^*},\, \U_{r^*}) \leq \tfrac{m_h}{A_h}\} = \emptyset, \ind{r^* \geq 1.}
	\end{equation}
	On the other hand, \eqref{eq:rem:rnabox_UUL:U_reg} along with $m_h < M_h,\, h \in \H$, yield
	\begin{equation}
		s(\L_{r^*},\, \U_{r^*}) \geq \tfrac{M_h}{A_h} > \tfrac{m_h}{A_h}, \ind{h \in \U_{r^*},}
	\end{equation}
	and hence, \eqref{eq:th:rnabox:proof:L} reads
	\begin{equation}
		\label{eq:th:rnabox:proof:L_no_U}
		\{h \in \H \setminus \L_{r^*}:\, s(\L_{r^*},\, \U_{r^*}) \leq \tfrac{m_h}{A_h}\} = \emptyset, \ind{r^* \geq 1.}
	\end{equation}
	
	The proof of necessity in \eqref{eq:th:rnabox:proof:optcond_U} is immediate in view of \eqref{eq:rem:rnabox_UUL:U_reg}, 
	whilst sufficiency follows by contradiction. Indeed, let $r^*\ge 1$. Assume that the right-hand side of equivalence 
	\eqref{eq:th:rnabox:proof:optcond_U} holds and $h \not\in \U_{r^*}$. Then, in view of Remark \ref{rem:rnabox_LcupU}, 
	either $h \in \H \setminus (\L_{r^*} \cup \U_{r^*})$ and then from \eqref{eq:rem:rnabox_rna_x}
	\begin{equation}
		s(\L_{r^*},\, \U_{r^*}) < \tfrac{M_h}{A_h}, 
	\end{equation}
	a contradiction, or $h \in \L_{r^*}$ and then from \eqref{eq:th:rnabox:proof:optcond_L},  in view of $m_h < M_h,\, h \in \H$,
	\begin{equation}
		\label{eq:th:rnabox:proof:reg_Uc}
		s(\L_{r^*},\, \U_{r^*}) \leq \tfrac{m_h}{A_h} < \tfrac{M_h}{A_h},
	\end{equation}
	a contradiction.
		
   	\item[Vector $\x^*$ is a {\em vertex} allocation:]
	For $r^* = 1$, the only possibility is that $\U_{r^*} = \H,\, \L_{r^*} = \emptyset$. Then,
	\eqref{eq:th:rnabox:proof:optcond_LU_ver} is clearly met, while \eqref{eq:th:rnabox:proof:optcond_LUn_ver} follows from 
	{\footnotesize Step \ref{alg:rnabox:rna}} of the {\em RNABOX} after referring to \eqref{eq:th:optcond_one_sided:M}.
	Let $r^* \geq 2$. Then, by \eqref{eq:rem:rnabox_UUL:U} we have
	\begin{equation}
		\label{eq:th:rnabox:proof:ver_U}
		s(\L_{r^*-1},\, \U_{r^*-1}) \geq \tfrac{M_h}{A_h}, \ind{h \in \U_{r^*-1} \supseteq \U_{r^*},}
	\end{equation}
	where the set inclusion is due to Lemma \ref{lem:U}. 
	On the other hand, from \eqref{eq:rem:rnabox_UUL:L}, we get
	\begin{equation}
		\label{eq:th:rnabox:proof:ver_L}
		s(\L_{r^*-1},\, \U_{r^*-1}) \leq \tfrac{m_h}{A_h}, \ind{h \in \L_{r^*-1} \cup \Lt_{r^*-1} = \L_{r^*},}
	\end{equation}
	where the fact that the above inequality is met for $h \in \L_{r^*-1}$ follows from \eqref{eq:lem:s_in_rnabox_any}. By 
	comparing \eqref{eq:th:rnabox:proof:ver_U} and \eqref{eq:th:rnabox:proof:ver_L} we clearly see that 
	\eqref{eq:th:rnabox:proof:optcond_LU_ver} is satisfied. Lastly, equation \eqref{eq:th:rnabox:proof:optcond_LUn_ver} is 
	fulfilled due to
	\begin{equation}
		n - m_{\Lt_1} - \ldots - m_{\Lt_{r*-1}} = n - m_{\L_{r^*}} = M_{\U_{r^*}},
	\end{equation}
	where the first equality follows from \eqref{eq:rnabox_Lsum} while the second one follows from {\footnotesize Step 
	\ref{alg:rnabox:rna}} of the {\em RNABOX}  after referring to \eqref{eq:th:optcond_one_sided:M} and using the fact that 
	$\U_{r^*} = \H \setminus \L_{r^*}$ for a {\em vertex} allocation.
\end{enumerate}

\section{Appendix: Optimality conditions for Problem \ref{prob:upper}}
\label{app:rna}

The following Theorem \ref{th:optcond_one_sided} provides necessary and sufficient conditions for the optimal solution to 
Problem \ref{prob:upper}. It was originally given as Theorem 1.1 in \citet{WWW} and it is crucial for the proof of Theorem 
\ref{th:rnabox}. Here, we will quote it in a slightly expanded form so that it also covers the case of $\U^* = \H$. As usual, the
set function $s$ is defined as in Definition \ref{def:s}. The algorithm that solves Problem \ref{prob:upper} is {\em RNA} and it is
given in Section \ref{sec:rnabox} of this paper.

\begin{theorem}
    \label{th:optcond_one_sided}
    The optimization Problem \ref{prob:upper} has a unique optimal solution. Point $\x^* = (x_h^*,\, h \in \H) \in \R_+^{\card{\H}}$ is a solution to optimization Problem \ref{prob:upper} if and only if $\x^*$ has entries of the form
    \begin{equation}
		x_h^* =
		\begin{cases}
			M_h,									& \ind{h \in \U^*} \\
			A_h\, s(\emptyset,\, \U^*),	& \ind{h \in \H \setminus \U^*,}
		\end{cases}
    \end{equation}
    with $\U^* \subseteq \H$, such that one of the following two cases holds:
    \begin{enumerate}[wide, labelindent=0pt, leftmargin=*]
		\item[CASE I:] $\U^* \subsetneq \H$ and
		\begin{equation}
			\U^* = \left\{h \in \H:\, A_h\, s(\emptyset,\, \U^*) \ge M_h \right\}.
		\end{equation}	
		\item[CASE II:] $\U^* = \H$ and
		\begin{equation}
			\label{eq:th:optcond_one_sided:M}
			n = \sum_{h \in \H} M_h.
		\end{equation}
    \end{enumerate}
\end{theorem}

\section{Appendix: Convex optimization scheme and the KKT conditions} 
\label{app:kkt}

A convex optimization problem is an optimization problem in which the objective function is a convex function and the feasible set is a convex set. In standard form it is written as
\begin{equation}
    \label{prob:convex}
    \begin{split}
		\underset{\x\, \in\, \D}{\texteq{minimize ~\,}} & \quad f(\x) \\
		\texteq{subject ~ to}  & \quad w_i(\x) = 0, \quad {i = 1, \ldots, k} \\
		& \quad g_j(\x) \le 0, \quad{j = 1, \ldots, \ell,}
    \end{split}
\end{equation}
where $\x$ is the optimization variable, $\D \subseteq \R^p,\, p \in \N_+$, the objective function $f:\D_f \subseteq \R^p 
\to \R$ and inequality constraint functions $g_j: \D_{g_j} \subseteq \R^p \to \R,\, j = 1, \ldots, \ell$, are convex, whilst equality 
constraint functions $w_i: \D_{w_i} \subseteq \R^p \to \R,\, i = 1, \ldots, k$, are affine. Here, $\D = \D_f \cap \bigcap_{i=1}^{k} 
\D_{w_i} \cap \bigcap_{j=1}^{\ell} \D_{g_j}$ denotes a common domain of all the functions. Point $\x \in \D$ is called {\em 
feasible} if it satisfies all of the constraints, otherwise the point is called {\em infeasible}. An optimization problem is called 
{\em feasible} if there exists $\x \in \D$ that is {\em feasible}, otherwise the problem is called {\em infeasible}.

In the context of the optimum allocation Problem \ref{prob} discussed in this paper, we are interested in a particular type of 
the convex problem, i.e. \eqref{prob:convex} in which all inequality constraint functions $g_j,\, j = 1, \ldots, \ell$, are affine. It 
is well known, see, e.g. the monograph \citet{Boyd}, that the solution for such an optimization problem can be identified 
through the set of equations and inequalities known as the Karush-Kuhn-Tucker (KKT) conditions, which in this case are not 
only necessary but also sufficient.
\begin{theorem}[KKT conditions for convex optimization problem with affine inequality constraints]
    A point $\x^* \in \D \subseteq \R^p,\, p \in \N_+$, is a solution to the convex optimization problem \eqref{prob:convex} in 
    which functions $g_j,\, j = 1, \ldots, \ell$, are affine if and only if there exist numbers $\lambda_i \in \R$, $i = 1, \ldots, k$, 
    and $\mu_j \geq 0$, $j = 1, \ldots, \ell$, called KKT multipliers, such that
    \begin{equation}
		\label{KKT}
		\begin{gathered}
			\nabla f(\x^*)+\sum_{i=1}^k \lambda_i \nabla w_i(\x^*) + \sum_{j=1}^\ell \mu_j \nabla g_j(\x^*) = \zero \\
			w_i(\x^*) = 0, \ind{i = 1, \ldots, k} \\
			g_j(\x^*) \le 0, \ind{j = 1, \ldots, \ell} \\
			\mu_j g_j(\x^*) = 0, \ind{j = 1, \ldots, \ell.}
		\end{gathered}
    \end{equation}
\end{theorem}

\bibliography{rnabox_rev_final_arxiv}
\end{document}